\documentclass[a4paper, reqno, 11pt]{amsart}        

\usepackage{amssymb}

\usepackage{epsfig}  		
\usepackage{epic,eepic}       

\let\oldtocsection=\tocsection
 
\let\oldtocsubsection=\tocsubsection
 
\let\oldtocsubsubsection=\tocsubsubsection
 
\renewcommand{\tocsection}[2]{\hspace{0em}\oldtocsection{#1}{#2}}
\renewcommand{\tocsubsection}[2]{\hspace{1em}\oldtocsubsection{#1}{#2}}
\renewcommand{\tocsubsubsection}[2]{\hspace{2em}\oldtocsubsubsection{#1}{#2}}

\usepackage{amsfonts}
\usepackage{amsthm}
\usepackage{amsmath}
\usepackage{amscd}
\usepackage[latin2]{inputenc}
\usepackage{t1enc}
\usepackage[mathscr]{eucal}
\usepackage{indentfirst}
\usepackage{graphicx}
\usepackage{graphics}
\usepackage{pict2e}
\usepackage{epic}
\numberwithin{equation}{section}
\usepackage[margin=2.9cm]{geometry}
\usepackage{hyperref}
\usepackage{dsfont}
\usepackage{csquotes}
\usepackage{stmaryrd}
\usepackage{blindtext}
\usepackage{tikz-cd}
\usepackage{mathrsfs}
\usepackage{cite}
\usepackage{mathtools}
\usepackage{extarrows}
\usepackage{epstopdf}
\usepackage[T1]{fontenc}
\usepackage{braket}

\theoremstyle{definition}
\newtheorem{definition}[equation]{Definition}
\newtheorem{example}[equation]{Example}
\newtheorem{proposition}[equation]{Proposition}
\newtheorem{theorem}[equation]{Theorem}
\newtheorem{remark}[equation]{Remark}

\newtheorem{lemma}[equation]{Lemma}

\numberwithin{equation}{section}

\newcommand{\midwedge}{\text{\Large$\wedge$}}
\newcommand{\midodot}{\text{\Large$\odot$}}
\newcommand{\LC}{\textrm{\tiny\tt LC}}
\newcommand{\dee}{\textrm{\tiny D}}

\newcommand{\be}{\begin{equation}}
\newcommand{\ee}{\end{equation}}
\def\beqa{\begin{eqnarray}}
\def\eeqa{\end{eqnarray}}
\def\bean{\begin{eqnarray*}}
\def\eean{\end{eqnarray*}}

\newcommand{\R}{\mathbb{R}}

\newcommand{\de}{\mathrm{d}}

\newcommand{\del}{\partial}

\newcommand{\IZ}{\mathbb{Z}}

\newcommand{\IR}{\mathbb{R}}

\newcommand{\IT}{\mathbb{T}}

\newcommand{\IX}{\mathbb{X}}
\newcommand{\IV}{\mathbb{V}}

\newcommand{\frh}{\mathfrak{h}}

\def\e{{\,\rm e}\,}

\newcommand{\cW}{{\mathcal W}}
\newcommand{\cN}{{\mathcal N}}

\newcommand{\cS}{{\mathcal S}}

\newcommand{\cH}{{\mathcal H}}

\newcommand{\cQ}{{\mathcal Q}}

\newcommand{\cL}{{\mathcal L}}
\newcommand{\cF}{{\mathcal F}}

\newcommand{\cK}{{\mathcal K}}

\newcommand{\cU}{{\mathcal U}}

\newcommand{\ccW}{{\mathscr W}}

\newcommand{\ccL}{{\mathscr L}}

\newcommand{\ccK}{{\mathscr K}}

\newcommand{\sfa}{{\mathsf{a}}}

\newcommand{\sfH}{{\mathsf{H}}}
\newcommand{\sfG}{{\mathsf{G}}}

\newcommand{\sfp}{{\mathsf{p}}}
\newcommand{\sfGamma}{{\mathsf{\Gamma}}}

\newcommand{\unit}{\mathds{1}}   			

\begin{document}

\title[D-Branes in Para-Hermitian Geometries]{D-Branes in Para-Hermitian Geometries}

\author[V.~E.~ Marotta]{Vincenzo Emilio Marotta}
\address[Vincenzo Emilio Marotta]
{Department of Mathematics and Maxwell Institute for Mathematical
  Sciences\\ Heriot-Watt
  University\\ Edinburgh EH14 4AS\\ United Kingdom}
\email{vm34@hw.ac.uk}

\author[R.~J. Szabo]{Richard J.~Szabo}
  \address[Richard J.~Szabo]
  {Department of Mathematics, Maxwell Institute for Mathematical Sciences and Higgs Centre for Theoretical Physics\\
  Heriot-Watt University\\
  Edinburgh EH14 4AS \\
  United Kingdom}
  \email{R.J.Szabo@hw.ac.uk}

\vfill

\begin{flushright}
\footnotesize
{\sf EMPG--22--01}
\normalsize
\end{flushright}

\vspace{1cm}

\begin{abstract}
We introduce T-duality invariant versions of D-branes in doubled geometry using a global covariant framework based on para-Hermitian geometry and metric algebroids. We define D-branes as conformal boundary conditions for the open string version of the Born sigma-model, where they are given by maximally isotropic vector bundles which do not generally admit the standard geometric picture in terms of submanifolds. When reduced to the conventional sigma-model description of a physical string background as the leaf space of a foliated para-Hermitian manifold, integrable branes yield D-branes as leaves of foliations which are interpreted as Dirac structures on the physical spacetime. We define a notion of generalised para-complex D-brane, which realises our D-branes as para-complex versions of topological A/B-branes. We illustrate how our formalism recovers standard D-branes in the explicit example of reductions from doubled nilmanifolds.
\end{abstract}

\maketitle

\begin{center}
{\sl\small Contribution to the Special Issue of Universe on `Dualities and Geometry'}
\end{center}

{\baselineskip=12pt
\tableofcontents
}


\section{Introduction}

This paper is concerned with the global descriptions of duality symmetric extended geometries and field theories. We focus on the manifestly T-duality invariant extension of the bosonic sector of type~II supergravity to double field theory~\cite{HullZw2009,Hull:2009zb}, which is based on the doubled formalism for string theory~\cite{Siegel1993a,Siegel1993b} (see e.g.~\cite{Berman2013} for a review). We work in the setting of para-Hermitian geometry and metric algebroids, which encompasses the standard flat space treatments (recovered as instances of para-K\"ahler geometry) and known examples of global doubled geometries (typically involving non-integrable almost para-complex structures). This approach was originally proposed by Vaisman~\cite{Vaisman2012,Vaisman2013}, and has recently undergone a renewed flurry of activity, see e.g.~\cite{Freidel2017,Jonke2018,Svoboda2018,Osten2018,Kokenyesi:2018xgj,Freidel2019,SzMar,Mori2019,Chatzistavrakidis:2019huz,Hassler:2019wvn,Hu:2019zro,Marotta:2019eqc,Sakatani:2020umt,Carow-Watamura:2020xij,Ikeda:2020lxz,Mori:2020yih,Grewcoe:2020gka,Marotta:2021sia} for a host of developments in this direction.

Given the pertinence of D-branes to non-perturbative aspects and as probes of short-distance geometry in string theory, as well as their role in gauge theories, we pose the question: What is the proper definition of a D-brane in the context of para-Hermitian geometry?  This question has not yet been tackled in full generality, and a detailed answer should lead to a natural incorporation of phenomena described by T-duality. 
The purpose of this contribution is to fill this gap. We attack this problem from two complementary perspectives, starting from the worldsheet point of view where D-branes are regarded as boundary conditions for two-dimensional non-linear sigma-models. D-branes in doubled geometry have been previously studied from this perspective on doubled torus bundles~\cite{Hull2005,Lawrence:2006ma,Hull2007}, on doubled twisted tori~\cite{Albertsson2009,Hull:2019iuy} and on doubled flat space~\cite{Albertsson:2011ux}; they have also been discussed in the U-duality invariant setting of exceptional target space geometries~\cite{Blair:2019tww}.

A para-Hermitian manifold endowed with a Born metric and its compatible metric algebroid can be used to define a doubled string sigma-model in a duality symmetric formulation. In~\cite{Marotta:2019eqc} we called this the `Born sigma-model'. The Born sigma-model is invariant under rigid T-duality transformations, and it can be viewed as a covariant extension of the metastring sigma-model into string phase spaces~\cite{Freidel2015}. In this sense it is a direct generalisation of Hull's doubled sigma-models into doubled torus bundles with manifest worldsheet covariance~\cite{Hull2005,Hull2007}. In this paper we present a new extension of the Born sigma-model by a Wess-Zumino term which encodes topologically non-trivial generalised NS--NS fluxes of the doubled geometry, thus generalising the doubled sigma-models into doubled twisted tori introduced in~\cite{Hull2009}.

The para-Hermitian geometry approach to double field theory is largely a theory of foliated manifolds. The choice of an almost para-complex structure represents a choice of polarisation of the doubled geometry, and a maximally isotropic foliation arises as a solution to the section constraint of double field theory~\cite{Vaisman2012}. Physical sigma-models are then obtained from the Born sigma-model by a Lie algebroid gauging along the foliation, which implements the quotient to the leaf space representing the physical spacetime. A para-Hermitian manifold typically admits different maximally isotropic Riemannian foliations, so that T-dual sigma-models are recovered on the respective leaf spaces. In this way the Born sigma-model unifies T-dual non-linear sigma-models in a geometric description.

In this paper we define the open string version of the Born sigma-model. D-branes are then boundary conditions for this two-dimensional field theory. Boundary conditions and D-branes for the Born sigma-model (as well as their extensions to branes in M-theory and IIB string theory) have also been discussed in a local formulation by~\cite{Sakatani:2020umt}, but only in the case of para-K\"ahler manifolds (which recovers conventional double field theory on flat space), after solving the section constraint and gauging with respect to a coordinate system adapted to the leaves of a foliation. Our analysis instead generalises the approach of~\cite{Albertsson2009}, which considered D-branes in doubled twisted tori, but also in a local setting and assuming integrability properties which are not implied by the boundary conditions. In the following we improve on these approaches by considering global aspects of D-branes as they arise solely from the boundary data of the Born sigma-model, without recourse to additional structures, which we instead discuss as supplementary input towards the more physical and geometrically intuitive notions of D-branes. We refer to the branes obtained in this way as (almost) `Born D-branes'.

From this perspective, D-branes in para-Hermitian geometry are seen as maximally isotropic vector bundles $L$. All branes are of the same rank, so that the physical D-branes of different dimensions are treated in a unified geometric setting. In particular, the splitting between tangent and normal directions, which are treated equally due to the manifest T-duality symmetry, is determined by $L$. However, no integrability condition is needed to solve the boundary conditions for the Born sigma-model, and even when it is imposed by hand, it generally only leads to a foliation of spacetime by the D-brane $L$. Thus these branes do not generally follow the standard geometric picture of a D-brane as a submanifold of spacetime. We discuss this point in some detail, as well as the related problem of how to incorporate D-branes carrying non-trivial gauge flux into this picture, which corresponds to bound states of D-branes of various dimensions. Even the `physical' D-branes, obtained by reducing Born D-branes to the leaf space, are generally only described as leaves of foliations of the physical spacetime.

To understand the geometrical characterisation of Born D-branes, we develop a complementary spacetime perspective on D-branes by exploiting the similarities between para-Hermitian geometry and generalised geometry~\cite{gualtieri:tesi,Hitchin2011}. These similarities are part of the crux of the para-Hermitian approach to doubled geometry, as they underlie the geometry behind the reduction of double field theory to supergravity upon solving the section constraint in a particular polarisation. In particular, the metric algebroids of doubled geometry should reduce to Courant algebroids. The ingredients comprising a Born geometry define an almost para-quaternionic structure, together with compatibility conditions with respect to the underlying metrics~\cite{Freidel2019}. These define a neutral almost hypercomplex structure, which restricts the type of T-dual string backgrounds that can be described in this framework (see~\cite{Kimura:2022dma} for a recent discussion). 

With these connections in mind, we extend Gualtieri's approach~\cite{gualtieri:tesi} to viewing abelian generalised complex D-branes as `generalised submanifolds' into the framework of metric algebroids and almost generalised para-complex structures. The setting provided by a generic metric algebroid turns out to be too coarse to provided any meaningful analysis, hence we restrict to the intermediary structure of a pre-Courant algebroid, which represents a structure along the way between the metric algebroids of double field theory and the Courant algebroids of supergravity arising from imposition of the section constraint~\cite{Jonke2018,Chatzistavrakidis:2019huz,Marotta:2021sia}. We refer to these branes as `generalised para-complex D-branes'. In particular, we show that the Born D-branes on an almost para-Hermitian manifold fit into this picture in a natural way: Any two-dimensional non-linear sigma-model naturally corresponds to an exact Courant algebroid (see e.g.~\cite{Alekseev:2004np,Severa2015,Severa-letters}); on an almost para-Hermitian manifold this is called the `large Courant algebroid'~\cite{Jonke2018,Chatzistavrakidis:2019huz,Hu:2019zro,Marotta:2021sia}. Seen in this way, our D-branes provide natural para-complex versions of the A-branes and B-branes of topological string theory. Moreover, using techniques of Courant algebroid reduction, this perspective yields an interpretation of the reduction of Born D-branes to the leaf space as Dirac structures on the physical spacetime.

We illustrate our formalism explicitly in the example of the standard six-dimensional doubled nilmanifold, which was also studied by~\cite{Albertsson2009,Hull:2019iuy}. We show how to recover the standard D-branes on the three-dimensional Heisenberg nilmanifold, and also on the $3$-torus $\IT^3$ with NS--NS $H$-flux; in particular, our formalism gives a natural geometric explanation of the well-known fact that D3-branes wrapping $\IT^3$ with $H\neq0$ are prohibited. It would be interesting to study how D-branes in other contexts are also realised by our framework, for instance the D-branes in Drinfel'd doubles which are related by Poisson-Lie T-duality~\cite{Klimcik:1996hp}, and the interplay between D-branes and non-abelian T-duality~\cite{Terrisse:2018hhf}. Both of these latter examples can be viewed as instances of the generalised T-duality in para-Hermitian geometry introduced by~\cite{Marotta:2019eqc}.

In this paper we discuss only bosonic string sigma-models, hence our discussion of D-branes is decoupled from issues of stability and preservation of supersymmetry. At the purely bosonic level, all that is needed is the boundary conformality condition,  and indeed our Born D-branes are conformally invariant. It would be interesting to develop supersymmetric extensions of the Born sigma-model, and to discern what new types of target space para-complex structures appear; (generalised) para-complex structures appear on target spaces for topologically twisted and supersymmetric sigma-models~\cite{Stojevic:2009ub,Kokenyesi:2018xgj,Hu:2019zro}, as well as in target spaces for $\cN=2$ hypermultiplets in Euclidean signature~\cite{Cortes:2005uq}. Another subtle issue we do not address in this paper is the extension to (multiple) non-abelian D-branes. In addition to the technical problems involved in the extension to higher rank gauge bundles, which are no longer intrinsic to the para-Hermitian geometry, it is not clear how to include the non-abelian transverse scalars which should now be sections of the tensor product of the normal bundle to the D-brane and the endomorphism bundle of the gauge bundle. A notion of higher rank generalised complex brane is discussed in~\cite{Gualtieri:2007ng}.

The remainder of this paper is structured into three parts. In Section~\ref{BSM} we review the Born sigma-model, extend it by a Wess-Zumino term and to open strings, derive and study properties of Born D-branes as boundary conditions for these two-dimensional field theories, and discuss how they induce D-branes on a leaf space representing a physical spacetime. In Section~\ref{sec:targetbranes} we provide putative spacetime notions of branes on a metric algebroid, introduce our notion of generalised para-complex D-branes showing how they naturally include the Born D-branes of Section~\ref{BSM} as special instances, discuss how they reduce to the physical spacetime, and finally present explicit examples on the doubled nilmanifold. To make the main messages of this paper clear, we have delegated all mathematical details surrounding para-Hermitian geometry and metric algebroids to Appendix~\ref{intropara} at the end of the paper; the reader unaquainted with these technical details may wish to consult there first before moving to the main body of the paper. \\[-2ex]

\paragraph{\bf Acknowledgments} \ 
 {\sc R.J.S.} would like to thank the editors Athanasios Chatzistavrakidis and Dimitrios Giataganas for the invitation to contribute to this special issue.
The work of {\sc V.E.M.} was funded by the STFC Doctoral Training Partnership Award ST/R504774/1 and a
Maxwell Institute Research Fellowship. The work of {\sc R.J.S.} was supported by
the STFC Consolidated Grant ST/P000363/1.

\section{Worldsheet Perspective: D-Branes on Para-Hermitian Manifolds} 
\label{BSM}

In this section we discuss string worldsheet sigma-models whose target space is a Born manifold, following and extending the treatment of~\cite{Marotta:2019eqc}. When extended to worldsheets with boundaries, an analysis of the open string boundary conditions leads to a global operational definition of a D-brane on a para-Hermitian manifold, generalizing the considerations and results for D-branes in doubled twisted tori from~\cite{Albertsson2009}. We describe several properties and special instances of our definition of D-brane, and particularly how they match with physical expectations from double field theory. 

\medskip

\subsection{Sigma-Models for Para-Hermitian Manifolds} ~\\[5pt]
We will start by defining sigma-models which describe harmonic maps from a worldsheet $\Sigma$, with metric $h$, into a para-Hermitian manifold $M$ using, when it exists, a generalised metric $\cH$ or a compatible Born geometry on $M$; see Appendices~\ref{sec:paravector} and~\ref{genmesubs} for the relevant background, definitions and notation. 

For this, we recall that the \emph{Dirichlet functional} $\cS_0:C^\infty(\Sigma,M)\longrightarrow\IR$ is obtained by endowing the space of maps $\de \phi:
T\Sigma \longrightarrow \phi^*TM$ for $\phi \in C^{\infty}(\Sigma, M)$ with a norm defined by
$\cH$, regarded as a metric on the vector space of sections of the pullback 
$\phi^*TM$ of the
tangent bundle $TM$ to $\Sigma$ by $\phi,$ and the inverse metric $h^{-1}$ on $T^*\Sigma$. This gives
a well-defined norm $\|\,\underline{\de\phi}\,\|_{h,\cH}$ for sections $\underline{\de \phi} \in
\mathsf{\Gamma}(T^*\Sigma \otimes \phi^*TM)$ which enables us to write
the Dirichlet functional as
\be \label{sigmanorm}
\cS_0[\phi]=\frac{1}{4}\,\int_{\Sigma}\, h^{-1}\big(\bar\cH (\de \phi, \de \phi)\big)\ \de \mu(h) =:
\frac{1}{4}\,\int_\Sigma\, \|\,\underline{\de \phi}\, \|_{h,\cH}^2\ \de \mu(h)\ ,
\ee 
where
$\bar \cH = \phi^* \cH$ is the pullback\footnote{Here and in the following a bar over a field on $M$ denotes its pullback to the worldsheet $\Sigma$ by the smooth map~$\phi:\Sigma\longrightarrow M$.} metric on $\phi^* TM,$ $\bar\cH (\de \phi, \de \phi) \in \mathsf{\Gamma}(T^* \Sigma \otimes T^*\Sigma)$ 
and
$$
\de \mu(h) = \star\,1
$$ 
is the area measure
induced by the worldsheet metric $h$ whose associated Hodge operator is denoted $\star$. For definiteness we take $h$ to be Lorentzian so that $\star^2=\mathds{1}.$

\begin{definition}
A \emph{sigma-model for a para-Hermitian manifold} is the theory of maps in $C^\infty(\Sigma,M)$ given by the action functional \eqref{sigmanorm}, where $\Sigma$ is a closed oriented surface
endowed with a Lorentzian metric $h$ and $(M,K, \eta)$
is an almost para-Hermitian manifold with a generalised metric $\cH$. A \emph{Born sigma-model} is a sigma-model for a para-Hermitian manifold where $\cH$ is a Born metric, i.e.~a compatible generalised metric in the sense of Definition~\ref{compagenmetr}.
\end{definition}

The equations of motion of the sigma-model are given by the equations for the critical points $\phi_0:(\Sigma,h)\longrightarrow(M,\cH)$ of \eqref{sigmanorm}. Their solutions are harmonic maps. The additional factors of $\frac12$ in \eqref{sigmanorm} and below ensure that the doubled sigma-model reduces to a conventional non-linear sigma-model on a `physical spacetime' after solving the section constraint~\cite{Marotta:2019eqc}.

The local expression for the Dirichlet functional \eqref{sigmanorm} is obtained as follows. Let $\cU \subset \Sigma$ be an open subset with a chart for $\Sigma$ consisting of local coordinates $(\sigma^\alpha) = (\tau, \sigma)$, and let $U \subset M$ be an open subset such that $\phi(\cU) \subseteq U.$ The local components of $\phi$ are given by $\phi^I \in C^{\infty}(\cU, \IR)$ for $I=1,\dots,2d$, where  $2d$ is the dimension of $M$. There is a frame $\{Z_I\}$ for $TU,$ i.e. a local frame for $TM$, such that $\underline{\de \phi}$ can be locally written as 
\be \label{localphi}
\underline{\de \phi}\,\big\rvert_{(\cU, U)}= \sum_{I=1}^{2d} \, \bar{Z}_I \otimes \de \phi^I =: \bar Z_I\otimes\de\phi^I \ ,
\ee
where $\bar{Z}_I= \phi^* Z_I$ is the pullback local frame for $\phi^*TM$, $\de$ is the de Rham differential on $\Sigma$ and we use the usual Einstein summation convention throughout. Then the action functional \eqref{sigmanorm} becomes
\be \label{sigmaaction}
\cS_0[\phi]=\frac{1}{4}\, \int_{\Sigma}\, \phi^*\big(\cH(Z_I, Z_J)\big) \, \de \phi^I \wedge \star\, \de \phi^J =: \frac14 \, \int_{\Sigma}\, \bar\cH_{IJ} \, \de \phi^I \wedge \star\, \de \phi^J\ .
\ee

We will also consider a topological term of the form 
\be \label{sigmatop}
\cS_{\rm top}[\phi]=\frac{1}{2}\, \int_{\Sigma}\, \bar\omega\ ,
\ee
which minimally couples the string to the fundamental 2-form $\omega$ of the almost para-Hermitian manifold $M.$ Its local expression is given by
\be \label{localtop} 
\cS_{\rm top}[\phi]=\frac{1}{4}\, \int_{\Sigma}\, \phi^*\big(\omega(Z_I, Z_J)\big) \, \de \phi^I \wedge  \, \de \phi^J  =: \frac{1}{4}\, \int_{\Sigma}\, \bar{\omega}_{IJ} \, \de \phi^I \wedge \, \de \phi^J \ .
\ee
This incorporates topologically trivial generalised fluxes on $M$ through the curvature 3-form
\begin{align*}
\cK = \de \omega \ .
\end{align*}

For curved worldsheets, the general form of a two-dimensional non-linear sigma-model also involves a
Fradkin-Tseytlin term
\be\nonumber
\cS_{\Psi}[\phi]=\frac1{4\pi}\, \int_{\Sigma}\,R^{(2)}(h) \, \bar\Psi \ \de
\mu(h) \ , 
\ee
where the smooth function $\Psi: M \longrightarrow \mathbb{R}$ is a scalar
dilaton field and $R^{(2)}(h)$ is the scalar curvature of the metric
$h$ on $\Sigma$. Since the metric $h$ is conformally equivalent to a
flat metric on $\Sigma$, this term can be (classically) set to $0$ by
a conformal transformation of the worldsheet and will not be
considered any further in the ensuing analysis.

We will denote by $\mathcal{S}(\cH, \omega)$ a Born sigma-model given by the sum of \eqref{sigmaaction} and~\eqref{localtop}, i.e.~the action functional is
\be \label{actcomplete}
 \cS[\phi]= \cS_0[\phi]+ \cS_{\rm top}[\phi] = \frac14 \, \int_{\Sigma}\, \bar\cH_{IJ} \, \de \phi^I \wedge \star\, \de \phi^J + \frac14\,\int_\Sigma\, \bar{\omega}_{IJ} \, \de \phi^I \wedge \, \de \phi^J  \ .
\ee
The notation stresses that the defining data for a Born sigma-model are given by
the fundamental geometric structures of a Born manifold. The action functional \eqref{actcomplete} is invariant under rigid ${\sf O}(d,d)$-transformations of any local frame $\{Z_I\}$ for $TM$, preserving the split signature metric $\eta$ (see Remark~\ref{rem:Tduality}), and so in this sense the Born sigma-model captures all T-dual sigma-models in a unified geometric description.

\medskip

\subsection{Wess-Zumino Terms from Canonical Metric Algebroids} ~\\[5pt]
If we wish to incorporate topologically non-trivial generalised fluxes in the doubled spacetime $M$, and correspondingly add a Wess-Zumino term in the Born sigma-model, we need some further restrictions on the underlying almost para-Hermitian manifold $(M,K,\eta)$. The inclusion of a canonical 3-form can be achieved by using metric algebroids on the tangent bundle $TM$ (which always exist, see Appendix~\ref{sec:malg}), together with Proposition~\ref{Dbracket3form} which shows that the difference between two D-brackets defining distinct metric algebroids with respect to the same split signature metric $\eta$ is indeed a 3-form. 
The almost para-Hermitian manifold admits a unique canonical compatible D-bracket by Example~\ref{ex:canmetricalg}, and we obtain a 3-form associated with it by choosing a reference metric algebroid. The most natural choice is induced by the Levi-Civita connection of $\eta$ (Example~\ref{ex:metricalgeeta}).

\begin{definition}
The \emph{canonical $3$-form} $H_{\texttt{can}}$ on an almost para-Hermitian manifold $(M,\!K,\!\eta)$ is given by
\begin{align}\label{eq:HfluxD}
H_{\textrm{\tt can}}(X,Y,Z) := \eta([\![X,Y]\!]_{\textrm{\tiny\tt LC}} - [\![X,Y]\!]_{\tt can},Z) \ ,
\end{align}
for all $X,Y,Z\in\sfGamma(TM)$. 
\end{definition}

The components of the canonical 3-form in the directions along the eigenbundles $L_\pm$ of the almost para-complex structure $K$ are given by
\begin{align*}
H_{\texttt{can}}(\sfp_\pm(X),&\,\sfp_\pm(Y),\sfp_\pm(Z)) \\[4pt]
& = \tfrac12\,\big( \eta([\sfp_\pm(X),\sfp_\pm(Y)],\sfp_\pm(Z)) \\ & \hspace{2cm}
- \eta([\sfp_\pm(X),\sfp_\pm(Z)],\sfp_\pm(Y)) + \eta(\sfp_\pm(X),[\sfp_\pm(Y),\sfp_\pm(Z)])\big) \ ,
\end{align*}
for all $X,Y,Z\in\sfGamma(TM)$, where $[\,\cdot\,,\,\cdot\,]$ is the Lie bracket of vector fields. These respectively vanish when $L_\pm$ is integrable. In particular, this gives a means for finding a relation between the components of $\cK=\de\omega$ and the Nijenhuis tensor of the almost para-complex structure $K$, see~\cite{Svoboda2018} for further details.

The canonical 3-form $H_{\texttt{can}}$ is in general neither closed nor has integer periods, hence we need to select a specific class of almost para-Hermitian manifolds which allow the introduction of a Wess-Zumino term in the Born sigma-model.

\begin{definition}\label{def:admissible}
An almost para-Hermitian manifold $(M,K,\eta)$ is \emph{admissible} if ${\sf H}_2(M)=0$ and (one half of) the canonical 3-form defines an integer cohomology class $\frac1{4\pi}\,[H_{\texttt{can}}]\in{\sf H}^3(M;\IZ)$.
\end{definition}

The condition on the degree~$2$ homology of $M$ guarantees that the closed image cycle $\phi(\Sigma)\subset M$ is a boundary for all maps $\phi\in C^\infty(\Sigma,M).$ The condition on the canonical 3-form is the \emph{Dirac quantisation condition} for the generalised fluxes, which implies geometrically that $H_{\texttt{can}}$ represents the Dixmier-Douady class of the tensor product of two stably isomorphic bundle gerbes on $M$ with connection. 

For an admissible almost para-Hermitian manifold, we can introduce an open three-dimensional manifold $V$ with boundary $\partial V=\Sigma$, and define the corresponding Wess-Zumino action functional $\cS_{\textrm{\tiny WZ}}:C^\infty(\Sigma,M)\longrightarrow \IR/2\pi\,\IZ$ by 
\begin{align}\label{eq:WZdouble}
\cS_{\textrm{\tiny WZ}}[\phi] = \frac12\,\int_V \, \phi^* H_{\tt can} \ ,
\end{align} 
where here we have smoothly extended the map $\phi\in C^\infty(\Sigma,M)$ to $V$. The  condition $\de H_{\tt can}=0$ will imply that the equations of motion for \eqref{eq:WZdouble} involve only the original map $\phi:\Sigma\longrightarrow M$. The Dirac quantisation condition ensures that, in the quantum theory, the contribution of the amplitude
$\exp(\,\mathrm{i}\,\cS_{\textrm{\tiny WZ}}[\phi])$ to the functional integral is well-defined,
i.e.~it is independent of the choice of three-dimensional manifold $V$ bounded by
$\Sigma$ and of the smooth extension of $\phi\in C^\infty(\Sigma,M)$ to $V$.

\medskip

\subsection{Equations of Motion and Boundary Conditions} ~\\[5pt]
In order to discuss open strings and D-branes in the context of the Born sigma-model, we assume henceforth that the two-dimensional worldsheet manifold $\Sigma$ has a non-empty boundary $\del \Sigma.$ Working throughout at lowest order in string perturbation theory, we take $\Sigma$ to be a flat worldsheet with metric $$h=-\de\tau\otimes\de\tau + \de\sigma\otimes\de\sigma \ , $$ and  choose the local coordinates $(\sigma^\alpha)=(\tau,\sigma)$ such that $\tau\in\IR$ is the local coordinate on the boundary $\del\Sigma$ while $\sigma\in\IR$ parametrises the normal direction.

The variational problem associated with the Born sigma-model can be formulated as follows. Let $\Sigma \times M$ be the total space of the (trivial) fibre bundle over $\Sigma$ with projection $\mathrm{pr}_\Sigma : \Sigma \times  M \longrightarrow \Sigma$, and let $J^1(\Sigma\times M)\longrightarrow \Sigma$ be the corresponding bundle of $1$-jets. Recall that $\mathsf{\Gamma}(\Sigma \times M) \simeq C^{\infty}(\Sigma, M)$ and $\underline{\de \phi} \in \mathsf{\Gamma}(T^*\Sigma \otimes \phi^* TM)$ for all $ \phi \in C^{\infty}(\Sigma, M).$ We shall now give an alternative definition of the action functional \eqref{actcomplete} for the Born sigma-model.

The Lagrangian function $\cL_0\in C^{\infty}\big(J^1(\Sigma \times M)\big)$ for the kinetic term \eqref{sigmanorm} is given by $\cL_0 (j^1 \phi)=\|\,\underline{\de\phi}\,\|_{h,\cH},$ where we identify the first jet prolongation $j^1\phi$ of a section $\underline{\phi}\in\mathsf{\Gamma}(\Sigma \times M)$ with $\underline{\de \phi}\,.$ The action functional is again of the form \eqref{sigmanorm}, where $\|\,\underline{\de\phi} \, \|_{h,\cH} \ \de \mu (h)$ is the Lagrangian density of the kinetic term.
The topological term \eqref{sigmatop} is included in this formulation by considering the Lagrangian density $(j^1\phi)^* \omega \in \mathsf{\Omega}^2(\Sigma).$

In the jet bundle formulation of classical field theory, the Euler-Lagrange equations for the total Lagrangian  $\cL$ are given by 
\begin{align}
(j^1\phi)^*\Big(\frac{\partial\cL}{\partial\IX^I} - \partial_\alpha\frac{\partial\cL}{\partial\IX_\alpha^I}\Big) = 0 \ ,
\end{align}
where $(\sigma^\alpha, \IX^I, \IX_\alpha^I)$ are local coordinates on the total space of the jet bundle $J^1(\Sigma \times M),$ and $\partial_\alpha:=\frac\partial{\partial\sigma^\alpha}$. For the Born sigma-model this gives
\be \label{eq:Borneom}
\de\star\,\bar\cH_{IJ}\,\de\phi^J +
\bar\cK_{IJK}\,\de \phi^J\wedge\de \phi^K=0 \ ,
\ee
where $\cK=\de\omega$.
Boundary conditions are obtained from the equations
\be \nonumber
(j^1 \phi)^* \cL \bigr\rvert_{\partial \Sigma}=  (j^1 \phi)^* \frac{\partial \cL}{\partial \IX_\alpha^I} \biggr\rvert_{\partial \Sigma}=0 \ , 
\ee
where $\underline{\phi} \in \mathsf{\Gamma}(\Sigma \times M)$ is a critical section of the Lagrangian $\cL$. 
Explicitly, the boundary conditions for the Born sigma-model read
\be \label{jetbou}
(j^1\phi)^*\big( -\tfrac{1}{2}\,\cH_{IJ}\, \IX_\sigma^J + \omega_{IJ}\, \IX_\tau^J \big) \big\rvert_{\partial \Sigma}=0 \ .
\ee

A similar form of the equations \eqref{jetbou} were obtained for the Born sigma-model in~\cite{Sakatani:2020umt}. We shall now show that they are the same as the equations which follow from the more standard approach to boundary conditions in sigma-models. These were obtained in a particular related context by~\cite{Albertsson2009}. 

The variation of the kinetic term of the Born sigma-model \eqref{sigmaaction} is given by
\be\nonumber
\delta \cS_0 = \frac{1}{4}\,\int_\Sigma\, \pounds_\varepsilon \big( \bar{\cH}_{IJ} \, \de \phi^I \wedge \star \, \de \phi^J \big) \ ,
\ee
where $\pounds_\varepsilon$ is the Lie derivative along the vector field $\varepsilon \in \sfGamma(T \Sigma)$ generating the variation.
We use the Cartan formula $\pounds_\varepsilon=\de\circ\iota_\varepsilon+\iota_\varepsilon\circ\de$, where $\iota_\varepsilon$ is the interior multiplication of forms by $\varepsilon$, together with Stokes\rq{} theorem. We then obtain the contribution to the boundary conditions as 
\be \label{kinvar}
\delta \cS_0 \big\rvert_{\partial \Sigma}= \frac{1}{2}\, \int_\Sigma\, \de \big((\iota_\varepsilon\, \de \phi^I)\, \bar{\cH}_{IJ} \star \de \phi^J\big) = \frac{1}{2}\,\int_{\partial\Sigma} \, (\iota_\varepsilon\, \de \phi^I)\, \bar{\cH}_{IJ} \star \de \phi^J \ .
\ee

The variation of the topological term \eqref{sigmatop} is similarly given by 
\be \label{vartop}
\delta \cS_{\rm top}= \frac{1}{2}\,\int_\Sigma\, \pounds_\varepsilon\, \bar{\omega}= \frac{1}{2}\,\int_\Sigma\, \iota_\varepsilon\, \bar{\cK} + \frac12\,\int_\Sigma\, \de\, \iota_\varepsilon\, \bar{\omega} \ .
\ee
The first term in the second equality of \eqref{vartop} contributes to the equations of motion, while the second term appears in the boundary conditions after using Stokes\rq{} theorem. This gives
\be \label{vartopcan}
\delta \cS_{\rm top}\big\rvert_{\partial \Sigma}= \frac{1}{2}\,\int_{\partial\Sigma} \, \iota_\varepsilon\, \bar{\omega} \ .
\ee

By adding the two variations \eqref{kinvar} and \eqref{vartopcan} together, and demanding that $\delta \cS\big|_{\partial\Sigma}=0$ for arbitrary variational vector fields $\varepsilon\in\sfGamma(T\Sigma)$, we obtain the boundary conditions \eqref{jetbou} in the more canonical form
\be \label{finalvar}
\big(-\tfrac{1}{2}\,\bar{\cH}_{IJ}\, \partial_\sigma \phi^J\, \de \sigma + \bar{\omega}_{IJ}\,\partial_\tau \phi^J\, \de \tau \big)\big\rvert_{\partial \Sigma} = 0 \ .
\ee
This generalises the boundary conditions which appeared in the approach of~\cite{Albertsson2009} to the doubled sigma-model for doubled twisted tori, and in~\cite{Sakatani:2020umt} for the Born sigma-model. From \eqref{finalvar} it follows that
\begin{align} \label{eq:conformalinv}
\partial_\tau \phi^I\,\bar{\cH}_{IJ}\, \partial_\sigma \phi^J\big\rvert_{\partial \Sigma} = 0 \ .
\end{align}
The left-hand side is the restriction to $\partial\Sigma$ of the component of the worldsheet energy-momentum tensor $T_{\tau\sigma}$ derived from \eqref{actcomplete}, and \eqref{eq:conformalinv} implies that it is conserved. In other words, the boundary conditions \eqref{finalvar} preserve conformal invariance of the Born sigma-model.

\medskip

\subsection{D-Branes for Born Sigma-Models} ~\\[5pt]
\label{sec:Dbranes}
Analysing the conformal boundary conditions \eqref{finalvar} leads to a definition of D-branes for the Born sigma-model, generalizing the approaches of~\cite{Albertsson2009,Sakatani:2020umt}.
For this, we require the solution of the equations of motion to satisfy the self-duality constraint~\cite{Marotta:2019eqc, Hull2005, Hull2007, Hull2009}
\be \label{selfduality}
\de \phi = \bar{\eta}^\sharp \circ \bar{\cH}^\flat(\star \, \de \phi) \ ,
\ee
where $\bar{\eta}^\sharp= \phi^*(\eta^\sharp)$ and $\bar{\cH}^\flat= \phi^*(\cH^\flat)$ are the pullbacks of the musical isomorphisms induced by $\eta$ and $\cH$, respectively. This constraint arises from the Born sigma-model with a Lie algebroid gauging that reduces it to the ordinary non-linear sigma-model on a `physical' target space~\cite{Marotta:2019eqc}. It enables us to provide solutions of the boundary conditions with a geometric interpretation. Written in this way,
the Born sigma-model \eqref{actcomplete} is an immediate generalisation of the sigma-models
for doubled torus fibrations that were introduced
in~\cite{Hull2007}. 

A solution of \eqref{finalvar} is given by a distribution $L \subset TM.$ To characterise the distribution $L$, we consider the short exact sequence of vector bundles on $M$ given by
\be \nonumber
0 \longrightarrow L \longrightarrow TM \longrightarrow \nu(L) \longrightarrow 0 \ ,
\ee
where $\nu(L) = TM / L$ is the normal bundle of $L.$ We then choose an orthogonal splitting $TM = L \oplus L^\perp$ with respect to the generalised metric $\cH$. This determines orthogonal projectors $\Pi \colon TM \longrightarrow L$ and $\Pi^\perp \colon TM \longrightarrow L^\perp .$

In the local form \eqref{localphi} for $\de \phi,$ the boundary conditions \eqref{finalvar} in the splitting $TM=L\oplus L^\perp$ are solved by setting
\begin{align}
\phi^*\big( \Pi^\perp (Z_I) \big) \otimes \del_\tau \phi^I \, \de \tau &=0 \ , \label{solution1}\\[4pt]
 \tfrac12\,\phi^*\big(\cH(\Pi(Z_I), Z_J)\big)\, \del_\sigma \phi^J\, \de \sigma &= \phi^*\big(\omega(\Pi(Z_I), Z_J)\big)\, \del_\tau \phi^J \, \de \tau \ . \label{solution2}
\end{align}  
On the other hand, the local form of the self-duality constraint \eqref{selfduality} is given by
\begin{align}
\bar{Z}_I \otimes \del_\tau \phi^I \, \de \tau &=  \bar{\eta}^\sharp \circ \bar{\cH}^\flat (\bar{Z}_J) \otimes \star \, \del_\sigma \phi^J\, \de \sigma \ , \label{selfdual1} \\[4pt]
\bar{Z}_I \otimes \del_\sigma \phi^I\, \de \sigma &=  \bar{\eta}^\sharp \circ \bar{\cH}^\flat(\bar{Z}_J) \otimes \star \, \del_\tau \phi^J\, \de \tau \ . \label{selfdual2}
\end{align} 
By combining \eqref{finalvar}, \eqref{solution2} and \eqref{selfdual2} we obtain
\be\nonumber
\phi^*\big(-\tfrac{1}{2}\, \eta (\Pi(Z_I), \Pi(Z_J)) + \omega(\Pi(Z_I), \Pi(Z_J)) \big)\, \partial_\tau \phi^J\, \de \tau =0 \ ,
\ee
which is solved by taking
\be \nonumber
\eta \big(\Pi(Z_I), \Pi(Z_J)\big)=0 \qquad \mbox{and} \qquad \omega\big(\Pi(Z_I), \Pi(Z_J)\big)=0 \ . 
\ee
These equations imply that the distribution $L$ must be isotropic with respect to both $\eta$ and $\omega$ in order to solve the boundary conditions \eqref{finalvar}.
 
By substituting \eqref{selfdual1} into \eqref{solution1} we obtain
\be \nonumber
\phi^*\big( \Pi^\perp ( \eta^\sharp \circ \cH^\flat(Z_I)) \big) \otimes \del_\sigma \phi^I\, \de \sigma =0 \ .
\ee
This equation together with \eqref{solution2} gives 
\be \nonumber
\eta \big(\Pi^\perp(Z_I), \Pi^\perp(Z_J)\big)=0 \ ,
\ee
where we used the isotropy of $L$ with respect to $\omega$ as well. It follows that $L^\perp$ is isotropic with respect to $\eta$, and that the ranks of both $L$ and $L^\perp$ are maximal, i.e. ${\rm{rank}}(L)= {\rm{rank}}(L^\perp)= d.$ Thus there are equal numbers of Neumann boundary conditions, given by \eqref{solution2}, and Dirichlet boundary conditions, given by \eqref{solution1}, as necessitated by T-duality. These isotropy conditions are consistent with conformal invariance: Substituting the self-duality constraints \eqref{selfdual1} and \eqref{selfdual2} into \eqref{eq:conformalinv}, and using the relation between the metrics $\eta$ and $\cH$ from Example~\ref{ex:etacHrel}, we find
\begin{align*}
\partial_\sigma\phi^I \, \bar\eta_{IJ} \, \partial_\sigma\phi^J\big|_{\partial \Sigma} = 0 \qquad \mbox{and} \qquad \partial_\tau\phi^I \, \bar\eta_{IJ} \, \partial_\tau\phi^J\big|_{\partial \Sigma} = 0 \ ,
\end{align*}
which shows that the restrictions of the sections $(j^1\phi)^*(\mathbb{X}_\sigma)$ and $(j^1\phi)^*(\mathbb{X}_\tau)$ to the boundary have images belonging to maximally isotropic sub-bundles of $\phi^*TM$.

\begin{definition} \label{dbranepara}
A \emph{D-brane for a Born sigma-model}, or \emph{almost Born D-brane} for short, is a maximally isotropic  distribution $L_{\textrm{\tiny D}} \subset TM$ which is preserved by the almost para-complex structure $K$, i.e. $K(L_{\textrm{\tiny D}}) = L_{\textrm{\tiny D}}.$ A \emph{Born D-brane} is an almost Born D-brane $L_{\textrm{\tiny D}}$ which is involutive, i.e. $[\sfGamma(L_{\textrm{\tiny D}}),\sfGamma(L_{\textrm{\tiny D}})]\subseteq\sfGamma(L_{\textrm{\tiny D}})$.
\end{definition}

\begin{remark}\label{rem:maxisotropic}
Definition~\ref{dbranepara} implies that $L_{\textrm{\tiny D}}\subset TM$ is a sub-bundle of constant rank, and moreover that $L_{\textrm{\tiny D}}$ is maximally isotropic with respect to both the split signature metric $\eta$ and the fundamental 2-form $\omega$, because of the invariance condition with respect to $K.$
\end{remark}

\begin{example}\label{ex:simplestDbranes}
Let $(M,K,\eta)$ be any almost para-Hermitian manifold. Then the eigenbundles $L_\pm$ of the almost para-complex structure $K\in{\sf Aut}_\unit(TM)$ are always almost Born D-branes.
\end{example}

\medskip

\subsection{Lagrangian Subspaces and Linear D-Branes} ~\\[5pt]
\label{sec:Lagrangian}
To work towards a better understanding of the general structure of a Born D-brane, we will first unravel the meaning of Definition~\ref{dbranepara} in the simple finite-dimensional setting where the target para-Hermitian vector bundle is the generalised tangent bundle of a point (see Example~\ref{ex:gentanbun}). We start by recalling the general construction of Lagrangian subspaces of $\IV=V\oplus V^*,$ where $V$ is a $d$-dimensional real vector space and $V^*$ is its linear dual. 

In our framework, $(\IV, K, \braket{\,\cdot\,,\, \cdot\,})$ is a para-Hermitian vector space with para-complex structure $K \in {\sf{Aut}}(\IV)$ defined by $K|_V=\unit_V$ and $K|_{V^*}=-\unit_{V^*}$, and split signature inner product induced by the duality pairing 
$$\braket{v+ \alpha, w + \beta}= \iota_v \beta+ \iota_w \alpha \ ,$$
where $v,w\in V$ and $\alpha,\beta\in V^*$.
We denote by ${\rm pr}_V: \IV \longrightarrow V$ the projection onto $V,$ i.e.~${\rm pr}_V(v+\alpha)=v$ for all $v+\alpha \in \IV.$

\begin{remark} \label{can2form}
By construction, $K$ and $\braket{\,\cdot\,,\, \cdot\,}$ satisfy the compatibility condition 
$$\braket{K(v+\alpha), K(w+\beta)}:=-\braket{v+\alpha, w+\beta} \ ,$$
for all $v+\alpha, \, w+\beta \in \IV$.
Thus there is a canonical 2-form on $\IV$ defined by 
$$\omega(v+\alpha, w+\beta):=\braket{K(v+\alpha), w+\beta}= \iota_v \beta- \iota_w \alpha \ .$$ 
\end{remark}

\begin{definition}
A \emph{Lagrangian subspace} of $(\IV, K, \braket{\,\cdot\,,\, \cdot\,})$ is a maximally isotropic vector subspace $L$ of $\IV$ with respect to the inner product $\braket{\,\cdot\,,\, \cdot\,}.$ We denote by ${\sf Lag}(\IV)$ the set of all Lagrangian subspaces of $\IV$.
\end{definition}

Notice that ${\sf Lag}(\IV)$ is non-empty, as $V, V^* \in {\sf{Lag}}(\IV).$

We now describe the construction of all Lagrangian subspaces. Let $W\subseteq V$ be a vector subspace of $V$ and $\Omega \in \midwedge^2\, W^*$ a 2-form on $W$. Let ${\sf Ann}(W)\subseteq V^*$ be the annihilator subspace of $W$ in $V^*$, i.e. the set of all $\alpha\in V^*$ such that $\iota_w\alpha=0$ for all $w\in W$. Then define 
\be \label{lagsub}
L(W, \Omega) \coloneqq  \{ v+ \iota_v \tilde{\Omega} + \alpha  \in \IV \ | \ v \in W \ , \ \alpha \in {\sf{Ann}}(W) \} \ \subset \ \IV
\ee
where $\tilde{\Omega}\in \midwedge^2\, V^*$ is any extension of $\Omega,$ i.e. $\tilde{\Omega}\rvert_{W}=\Omega,$ or equivalently $\iota_w\iota_v \tilde{\Omega}= \iota_w\iota_v \Omega$ for all~$v,w \in W.$ We have the well-known result of

\begin{proposition}
$L(W, \Omega) \subset \IV$ is a Lagrangian subspace of $\IV$ which is independent of the choice of extension of $\Omega.$
\end{proposition}

We can show that any Lagrangian subspace of $\IV$ has the form $L(W, \Omega)$ through

\begin{theorem} \label{lagwo}
For any $L \in \sf{Lag}(\IV)$, define the vector subspace $W \coloneqq \mathrm{pr}_V (L)\subseteq V$ and the 2-form $\Omega \in \midwedge^2\, W^*$ by $\Omega(v,w) \coloneqq \iota_w \alpha=-\iota_v\beta, $ for all
$v+\alpha, w+\beta \in L.$ Then $L=L(W,\Omega)$.
\end{theorem}
\begin{proof}
Since $W \coloneqq \mathrm{pr}_V (L),$ it follows that $\ker\big(\mathrm{pr}_V \big\rvert_L\big)= {\sf{Ann}}(W).$ Thus there is a short exact sequence of vector spaces given by
\be \label{seqL}
0 \longrightarrow {{\sf{Ann}}(W)} \xlongrightarrow{ \ i \ } L \xlongrightarrow{\mathrm{pr}_V} W \longrightarrow 0 \ ,
\ee
where $i$ is the subspace inclusion. 
Let $s: W \longrightarrow L$ be a splitting of \eqref{seqL}. Then ${\rm{im}}(s)= \{ v + A(v) \in L\  |\ v \in W \}$ for a linear map $A: W \longrightarrow V^*.$
Therefore 
$$L={\rm{im}}(s)\oplus {\sf{Ann}}(W)=\{ v+ A(v)+ \alpha \in \IV\  | \ v \in W \ , \ \alpha \in  {\sf{Ann}}(W)\} \ .$$
Since $L$ is both isotropic and coisotropic with respect to $\braket{\,\cdot\,,\, \cdot\,}$, it follows that
$$0=\braket{v+A(v)+\alpha, w+A(w)+\beta}=\iota_v A(w)+ \iota_w A(v) \ .$$
In other words, the map $W \longrightarrow W^*$ given by $v \longmapsto A(v)\rvert_W$ is skew-symmetric and so determines a 2-form $\Omega$ on $W$ such that $\iota_v \Omega \coloneqq A(v)\rvert_W.$

Let $\tilde{\Omega} \in \midwedge^2\, V^*$ be any extension of $\Omega$ and define $$\gamma_v \coloneqq  A(v)- \iota_v \tilde{\Omega} \ \in \ {\sf{Ann}}(W) \ , $$
for all $v\in W$. 
Then any element $v+A(v)+\alpha \in L$ can be written as
$$v+A(v)+ \alpha= v+ \iota_v \tilde{\Omega} + \gamma_v +\alpha= v+\iota_v \tilde{\Omega} + \alpha\rq{} $$
where $\alpha\rq{}=\alpha+\gamma_v \in {\sf{Ann}}(W).$ Hence $L=L(W, \Omega).$
\end{proof}

With these preliminary considerations we can now formulate precisely what is meant by a D-brane in this setting. 

\begin{definition}\label{def:linearD}
A \emph{linear D-brane} in the para-Hermitian vector space $(\IV, K, \braket{\,\cdot\,,\, \cdot\,})$ is a Lagrangian subspace $L$ which is preserved by the para-complex structure $K\in {\sf{Aut}}(\IV),$ i.e.~$K(L) = L.$
\end{definition}

\begin{remark}
The invariance condition on $L$ in Definition~\ref{def:linearD} with respect to $K$, as well as its Lagrangian condition, together imply that $L$ is also maximally isotropic with respect to the 2-form $\omega$ introduced in Remark~\ref{can2form}.
\end{remark}

\begin{proposition} \label{lindbraneW}
Let $L$ be a linear D-brane in $(\IV, K, \braket{\,\cdot\,,\, \cdot\,})$, and let $W\subseteq V$ be the vector subspace $W:={\rm pr}_V(L)$. Then $L=W \oplus {\sf{Ann}}(W).$
\end{proposition}
\begin{proof}
Since $L\in{\sf Lag}(\IV)$, by Theorem~\ref{lagwo} it follows that $L=L(W,\Omega)$. We show that $\Omega=0$. 
For any $v+\iota_v \tilde{\Omega} +\alpha, \, w+\iota_w\tilde{\Omega}+\beta \in L(W, \Omega)$, we compute
\begin{align}
0=\omega(v+\iota_v \tilde{\Omega} +\alpha,  w+\iota_w\tilde{\Omega}+\beta)&= \braket{K(v+\iota_v \tilde{\Omega} +\alpha), w+\iota_w\tilde{\Omega}+\beta} \nonumber\\[4pt]
&=\braket{v -\iota_v \tilde{\Omega} -\alpha, w+\iota_w\tilde{\Omega}+\beta} \nonumber\\[4pt]
&=2\,\iota_v\iota_w\tilde{\Omega} \ . \nonumber
\end{align}
Thus $\iota_v \iota_w\Omega=0,$ for all $ v,w \in W.$
\end{proof}

\medskip

\subsection{General Structure of Born D-Branes} ~\\[5pt]
\label{sec:genBornDbranes}
We can now describe the general structure of a Born D-brane, beyond the simplest cases provided by Example~\ref{ex:simplestDbranes}, by applying the linear algebra developed in Section~\ref{sec:Lagrangian} fibrewise. Let $(M,K,\eta)$ be an almost para-Hermitian manifold. Recall that $TM=L_+ \oplus L_-$, where the eigenbundles $L_\pm$ of the almost para-complex structure $K$ are maximally isotropic sub-bundles of $TM$ with respect to $\eta.$ Denote by $\sfp_\pm : TM \longrightarrow L_\pm$ the projections onto $L_\pm\subset TM.$ If $L \subset TM$ is any sub-bundle of constant rank, then $$W_L \coloneqq \sfp_+(L)=L_+ \cap L$$ is a vector sub-bundle of $L_+.$ We denote by $$\overline{W}{\!}_L\coloneqq \eta^\sharp\big({\sf{Ann}}(W_L)\big) $$ the image in $L_-$ of the annihilator sub-bundle ${\sf{Ann}}(W_L) \subset T^*M$ under the musical isomorphism $\eta^\sharp : T^*M \longrightarrow TM$ induced by the split signature metric $\eta.$

\begin{lemma} \label{lagrconstrank}
Let $(M, K, \eta)$ be an almost para-Hermitian manifold and $L\subseteq TM$ a maximally isotropic sub-bundle such that $W_L\coloneqq \sfp_+(L)$ has constant rank. Then $L$ can be expressed in the form
$$L=\big\{ e+  \bar{\Omega}(e) + \bar{\alpha} \ \big| \ e \in W_L \ , \ \bar{\alpha} \in \overline{W}\!_L \big\} $$
where $\bar\Omega: W_L \longrightarrow L_- $ is a vector bundle morphism covering the identity whose restriction to $W_L$ induces a 2-form $\Omega \in \sfGamma(\midwedge^2\, W_L^*).$
\end{lemma}
\begin{proof}
Since $L$ is maximally isotropic, the restriction $\bar{\Omega}\rvert_{W_L}: W_L \longrightarrow \overline{W}\!_L$ is fibrewise skew-symmetric in the sense that $\Omega\coloneqq \eta^\flat \circ \bar\Omega\rvert_{W_L} \in \sfGamma(\midwedge^2 W^*_L)$. 
The result now follows from Theorem~\ref{lagwo} via the fibrewise identification $T M={L_+}\oplus {L_-} \simeq {L_+} \oplus L_+^*$ by using the metric $\eta.$
\end{proof}

\begin{proposition}\label{prop:LDWD}
Let $(M, K, \eta)$ be an almost para-Hermitian manifold and $L_{\textrm{\tiny D}} \subset TM$ an almost Born D-brane such that $W_{\textrm{\tiny D}}:=W_{L_{\textrm{\tiny D}}} = \sfp_+(L_{\textrm{\tiny D}})$ has constant rank. Then $$L_{\textrm{\tiny D}}=W_{{\textrm{\tiny D}}}\,\oplus\, \overline{W}\!_{{\textrm{\tiny D}}} \ . $$
\end{proposition}
\begin{proof}
This is a straightforward consequence of Proposition~\ref{lindbraneW} and Lemma~\ref{lagrconstrank}.
\end{proof}

\begin{remark}[\bf Worldvolumes]\label{rem:dbranesupport}
An almost Born D-brane $L_{\textrm{\tiny D}}$ gives a general solution to the boundary conditions for the Born sigma-model, which suffices for most purposes. However, it does not necessarily admit a geometric interpretation as a brane wrapping a submanifold of the target manifold $M$. This is reminiscent of the algebraic definition of D-branes as boundary states in abstract conformal field theory~\cite{Cardy:1989ir}, which in some cases also do not admit target space interpretations as open strings ending on worldvolumes, yet they are consistent boundary conditions for the two-dimensional field theory. 

The involutivity condition for a Born D-brane ensures that $L_{\textrm{\tiny D}}$ is Frobenius integrable, so that $L_{\textrm{\tiny D}}\simeq T\cF_{\textrm{\tiny D}}$ induces a foliation $\cF_{\textrm{\tiny D}}$ of the Born manifold $M.$ A leaf of the foliation $\cF_{\textrm{\tiny D}}$ is then interpreted physically as providing a geometric picture of a `D-brane worldvolume', whose dimension is ${\rm rank}(L_{\textrm{\tiny D}})=d$. The inclusion of all leaves accounts for the moduli of D-branes, i.e. the transverse displacements of the worldvolume in spacetime, as is necessary in any T-duality invariant formulation of D-branes. This is analogous to the considerations of certain D-branes as (singular) foliations in generalised geometry, which appears in e.g.~\cite{Zabzine:2004dp,Asakawa2012}. Locally, the integrability condition reads $[\Pi(Z_I),\Pi(Z_J)]\in\sfGamma(L_{\textrm{\tiny D}})$, which generalises the analogous condition in~\cite{Albertsson2009}.
\end{remark}

\begin{remark}[\bf Geometry] 
An almost Born D-brane $L_\dee$ naturally inherits a metric from the target para-Hermitian manifold. The Born metric $\cH$ given by \eqref{eq:diaghermmetric} is determined by a metric $g_+$ on the eigenbundle $L_+$, which induces a metric on $L_\dee$. Assume as previously that $W_\dee = L_\dee \cap L_+$ has constant rank. Then $W_\dee$ admits the restriction 
$g_\dee$ of the metric $g_+$ and $L_\dee=  W_\dee \oplus \overline{W}\!_\dee$ admits the metric  
\be \nonumber
\mathcal{H}_\dee = \bigg( \begin{matrix}
g_\dee & 0 \\ 0 & \eta^\flat(g_\dee^{-1})
\end{matrix} \bigg) \ ,
\ee
where we note that the ranks of $W_\dee$ and $\overline{W}\!_\dee$ are not generally equal.
\end{remark}

\begin{remark}[\bf Gauge Fields]\label{rem:B+F}
In the quantum theory, the massless states of the open string sigma-model should introduce gauge fields on the D-brane, but in the general setting of Definition~\ref{dbranepara} this is not so straightforward to describe geometrically.
In the case of a Born D-brane $L_{\textrm{\tiny D}}$ with induced foliation $\cF_{\textrm{\tiny D}}$, a field strength $F$ can be introduced by applying a $B_+$-transformation, as discussed in Appendix~\ref{sec:Btransformations} (a similar approach to incorporating gauge flux appears in the local approach of~\cite{Sakatani:2020umt} in the case of para-K\"ahler manifolds). Let $B_+ \colon L_+ \longrightarrow L_-$ be a vector bundle morphism over the identity with induced $2$-form $b_+\in\sfGamma(\midwedge^2L_+^*)$. The pullback of a D-brane \smash{$L_{\textrm{\tiny D}}^{B_+} = e^{-B_+}(L_{\textrm{\tiny D}})$} by the $B_+$-transformation is preserved by the pullback of the almost para-complex structure $K_{B_+}= e^{-B_+} \circ K \circ e^{B_+},$ but not by $K$ itself. In this new polarisation, the D-brane $L_{\textrm{\tiny D}}$ acquires a non-vanishing $2$-form $-\frac12\,\omega_{B_+}|_{L_{\textrm{\tiny D}}} = b_+|_{L_{\textrm{\tiny D}}}$, induced by the map $\bar{\Omega} = B_+|_{L_{\textrm{\tiny D}}}$ in Lemma~\ref{lagrconstrank}. 

By considering an embedded leaf $\cW_{\textrm{\tiny D}}\lhook\joinrel\longrightarrow M$ of $\cF_{\textrm{\tiny D}},$ the 2-form $F\in{\mathsf{\Omega}}^2(\cW_{\textrm{\tiny D}})$ is then given by
\be\nonumber
F \coloneqq b_+ \rvert_{\cW_{\textrm{\tiny D}}} \ .
\ee 
If $B_+$ induces a closed 2-form $b_+$ on $L_+$ whose restriction to $\cW_{\textrm{\tiny D}}$ has integer periods, then by Chern-Weil theory $F$ is the curvature $2$-form of a connection $\nabla^C$ on a complex line bundle $C\longrightarrow\cW_{\textrm{\tiny D}}$, which is interpreted physically as providing a geometric picture of a `Chan-Paton bundle' on the D-brane. Note that despite the requirement $\de b_+ = 0,$ the sub-bundle \smash{$L_{\textrm{\tiny D}}^{B_+}\subset TM$} still might not be integrable.

In an analogous way, one can introduce the `transverse scalar fields' to a Born D-brane $L_\dee$, which is given by a section $\varrho$ of the maximally isotropic normal bundle $\nu(L_\dee)\simeq L_\dee^\perp$. Then the pullback $L_\dee^\varrho=f_\varrho^*(L_\dee)$ by the diffeormorphism $f_\varrho:M\longrightarrow M$ generated by the vector field $\varrho$ is preserved by the corresponding pullback of the almost para-complex structure $K_\varrho = f_\varrho^*\circ K\circ f_\varrho^*{}^{-1}$.
\end{remark} 

\begin{remark}[{\bf T-Duality}] \label{rem:Tduality}
We can describe how an almost Born D-brane transforms under the \emph{generalised T-duality} discussed by~\cite{Marotta:2019eqc}. These transformations form the subgroup ${\sf O}(TM)\subset{\sf Aut}(TM)$ of tangent bundle automorphisms of the almost para-Hermitian manifold $(M,K,\eta)$ which preserve the split signature metric $\eta$; its elements are metric-preserving pairs $\vartheta=(\bar f,f)$ of a vector bundle isomorphism $\bar f:TM\longrightarrow TM$ covering a diffeomorphism $f:M\longrightarrow M$. Examples include isometric diffeomorphisms $f:M\longrightarrow M$, for which $\bar f=f^*$, and $B_+$-transformations, which cover the identity. The natural group of discrete transformations is
\begin{align*}
{\sf O}(TM;\IZ) := {\sf O}(TM) \, \cap \, {\sf Diff}(M;\IZ) \ ,
\end{align*}
where ${\sf Diff}(M;\IZ)\subset {\sf Diff}(M)$ is the subgroup of large diffeomorphisms of $M$. This generalises the usual T-duality group of torus bundles.

A generalised T-duality $\vartheta=(\bar f,f)$ pulls back a Born geometry $(K,\eta,\cH)$ on $M$ to the Born geometry $(K_\vartheta,\eta,\cH_\vartheta) = (\bar f\circ K\circ\bar f^{-1},\eta,\bar f^*\cH)$. A D-brane $L_\dee\subset TM$ for the Born sigma-model $\cS(\cH,\omega)$ into $(M,K,\eta)$ is then pulled back to the D-brane $L_\dee^\vartheta = \bar f(L_\dee)$ for the Born sigma-model $\cS(\cH_\vartheta,\omega_\vartheta)$ into $(M,K_\vartheta,\eta)$; indeed, if $L_\dee$ is maximally isotropic then so is $L_\dee^\vartheta$ (as $\vartheta$ is an isometry of $\eta$), while $K(L_\dee)=L_\dee$ implies $K_\vartheta(L_\dee^\vartheta) = L_\dee^\vartheta$. 
\end{remark}

\medskip

\subsection{Lagrangian Born D-Branes} ~\\[5pt]
\label{sec:LagrangianBorn}
A particular instance that naturally leads to a notion of `worldvolume' for an almost Born D-brane is when one wishes to consider D-branes in topologically non-trivial generalised flux backgrounds. When the worldsheet $\Sigma$ has a non-empty boundary $\partial\Sigma$, both the notion of admissibility from Definition~\ref{def:admissible} and the definition of the Wess-Zumino term \eqref{eq:WZdouble} require modification. Following~\cite{Figueroa-OFarrill:2005vws}, in this case we should reformulate the Born sigma-model as a theory of \emph{relative maps} $$\phi:(\Sigma,\partial\Sigma)\longrightarrow (M,\cW) \ , $$ where $\cW\subset M$ is a given fixed submanifold such that $\phi(\partial\Sigma)\subset \cW$. 

We now assume that the relative degree~$2$ homology of $(M,\cW)$ is trivial, ${\sf H}_2(M,\cW)=0$, and that there exists a $2$-form $B_{\texttt{can}}$ on $\cW$ such that the pair $(H_{\tt can},B_{\tt can})$ defines an integer relative cohomology class $\frac1{4\pi}\,[(H_{\tt can},B_{\tt can})]\in{\sf H}^3(M,\cW;\IZ)$. The former assumption ensures that the image chain $\phi(\Sigma)$ is a relative boundary modulo $\cW$ for all relative maps $\phi:(\Sigma,\partial\Sigma)\longrightarrow (M,\cW)$. The latter condition implies, in particular, that the canoncial $3$-form $H_{\tt can}$ is again closed and in addition that  its restriction to $\cW$ obeys
\begin{align}\label{eq:HBcan}
i^*H_{\tt can} = \de B_{\tt can} \ ,
\end{align}
where $i:\cW\lhook\joinrel\longrightarrow M$ is the embedding of $\cW$ in $M$. 

We can subsequently modify \eqref{eq:WZdouble} to the \emph{relative} Wess-Zumino action functional defined by~\cite{Figueroa-OFarrill:2005vws}
\begin{align}\label{eq:relativeWZdouble}
\cS_{\textrm{\tiny WZ}}[\phi] = \frac12\,\int_V\,\phi^*H_{\tt can} - \frac12\,\int_\Delta \, \phi^*B_{\tt can} \ ,
\end{align}
where now $V$ is a three-manifold with boundary $\partial V = \Sigma+\Delta$ such that $\phi(\Delta)\subset \cW$, and as before we have smoothly extended the relative map $\phi$ to $V$. By virtue of \eqref{eq:HBcan}, the canonical $3$-form contributes to the equations of motion \eqref{eq:Borneom} by shifting the curvature $\cK=\de\omega$  to $\cK+H_{\tt can}$, whereas the $2$-form $B_{\tt can}$ only contributes to the boundary conditions \eqref{finalvar} by shifting the fundamental $2$-form $\omega$ to $\omega+B_{\tt can}$. Indeed, using $\de H_{\tt can}=0$ and Stokes' theorem, the variation of \eqref{eq:relativeWZdouble} is computed to be
\begin{align*}
\delta \cS_{\textrm{\tiny WZ}} &= \frac12\,\int_V\,\pounds_\varepsilon\bar H_{\tt can} - \frac12\,\int_\Delta\,\pounds_\varepsilon\bar B_{\tt can} \\[4pt]
&= \frac12\,\int_{\partial V} \, \iota_\varepsilon\bar H_{\tt can} - \frac12\,\int_\Delta\,\iota_\varepsilon\de\bar B_{\tt can} - \frac12\,\int_{\partial\Delta}\,\iota_\varepsilon\bar B_{\tt can} \\[4pt]
&= \frac12\,\int_\Sigma \iota_\varepsilon\bar H_{\tt can} + \frac12\,\int_{\partial\Sigma}\,\iota_\varepsilon\bar B_{\tt can} \ ,
\end{align*}
where in the last step we used $\bar H_{\tt can}=\de \bar B_{\tt can}$ on $\Delta$ and $\partial\Delta=-\partial\Sigma$. The relative Dirac quantisation condition on $[(H_{\tt can},B_{\tt can})]$ guarantees that the functional integral is independent of the choice of pair~$(V,\Delta)$.

The addition of the Wess-Zumino term \eqref{eq:relativeWZdouble} also does not affect the self-duality constraint \eqref{selfduality}, so that the analysis of the boundary conditions \eqref{finalvar} proceeds exactly as before by replacing the fundamental $2$-form everywhere with $\omega+B_{\tt can}$. As $B_{\tt can}$ is defined only on $\cW\subset M$, in this case it is natural to choose the distribution
\begin{align*}
L_{\cW} := {\rm im}(\de i) \ ,
\end{align*}
where $\de i:T\cW\longrightarrow TM$ is the derivative of the embedding $i:\cW\lhook\joinrel\longrightarrow M$. Demanding that $L_\cW$ be an almost Born D-brane in the sense of Definition~\ref{dbranepara} then means that $\cW$ is a \emph{Lagrangian submanifold} of $M$ with respect to $\eta$, such that $K(L_\cW)= L_\cW$, and moreover (by Remark~\ref{rem:maxisotropic}) that
$
B_{\tt can}(\Pi(Z_I),\Pi(Z_J)) = 0 .
$
Since $B_{\tt can}$ is only defined on $\cW$, it follows that $B_{\tt can}=0$, and hence $i^*H_{\tt can}=0$ by \eqref{eq:HBcan}, or equivalently
\begin{align}\label{eq:Hcan0}
H_{\tt can}\big(\Pi(Z_I),\Pi(Z_J),\Pi(Z_K)\big) = 0 \ .
\end{align}
This is a generalisation of the `orientation' condition of~\cite{Albertsson2009} for D-branes in doubled twisted tori.

The middle-dimensional submanifold $\cW\subset M$ is regarded as the worldvolume of the almost Born D-brane $L_\cW\subset TM$, which we call a \emph{Lagrangian Born D-brane} in this case. It is the analogue in para-Hermitian geometry of the conventional A-branes (D-branes of the topological A-model) which are supported on Lagrangian submanifolds of a complex symplectic manifold~\cite{Witten:1992fb,Kapustin:2001ij}. The condition \eqref{eq:Hcan0} forbids Lagrangian Born D-branes whose worldvolumes support non-zero generalised fluxes. Note that this is a stronger requirement than the vanishing of the Freed-Witten anomaly~\cite{Freed:1999vc}, which would only require the canonical $3$-form $H_{\tt can}$ to become topologically trivial when restricted to $\cW$, as in \eqref{eq:HBcan}.

One virtue of dealing with Lagrangian Born D-branes is that it is straightforward to couple them to gauge fields through the introduction of Chan-Paton factors in the Born sigma-model. For this, let $C$ be a complex line bundle on the submanifold $\cW\subset M$ endowed with a unitary connection $\nabla^C$ whose curvature $2$-form is denoted $F$. Locally, $F=\de A$ where $A$ is the (local) connection $1$-form characterising $\nabla^C$. The string endpoint, which propagates on the boundary $\partial\Sigma$ of the worldsheet, is charged with respect to the gauge field $A$ on the brane. Since $\phi(\partial\Sigma)\subset \cW$, its incorporation into the Born sigma-model is achieved by adding the minimal coupling term
\begin{align}\label{eq:CP}
\cS_{\textrm{\tiny CP}}[\phi] := \frac12\,\int_{\partial\Sigma} \, \phi^*A = -\frac12\,\int_{\partial\Delta} \, \phi^*A = -\frac12\, \int_\Delta\,\phi^*F \ ,
\end{align}
where we used $\partial\Sigma=-\partial\Delta$ and Stokes' theorem. 

Combining \eqref{eq:CP} with \eqref{eq:relativeWZdouble} shows that the overall effect is to shift the $2$-form $B_{\tt can}$ to $B_{\tt can}+F$, which is also defined only on $\cW$. Repeating the arguments above shows that now
\begin{align}\label{eq:projflat}
F + B_{\tt can} = 0 \ .
\end{align} 
This implies that the $2$-form $B_{\tt can}$ defines an integer cohomology class $\frac1{2\pi}\,[B_{\tt can}]\in {\sf H}^2(\cW;\IZ)$, and so \eqref{eq:HBcan} again leads to the vanishing flux constraint \eqref{eq:Hcan0}. The condition \eqref{eq:projflat} means that Lagrangian Born D-branes can only couple to \emph{projectively flat} connections on $\cW$. If we choose to set $B_{\tt can}=0$ (which we may do as the role of $B_{\tt can}$ is superfluous at this stage), then this is again analogous to the case of Lagrangian A-branes, which necessarily come with flat Chan-Paton bundles~\cite{Witten:1992fb,Kapustin:2001ij}. 

\medskip

\subsection{D-Branes on the Leaf Space} ~\\[5pt]
\label{sec:Dbranesleaf}
We shall now discuss how Definition \ref{dbranepara} induces the canonical notion of a D-brane for the \lq\lq{}physical\rq\rq{} non-linear sigma-model obtained with the reduction procedure discussed in \cite{Marotta:2019eqc}.

Let $(M, K, \eta)$ be an almost para-Hermitian manifold with a Born metric $\cH$ such that the eigenbundle $L_-$ of $K$ is integrable, i.e. $L_-=T\cF$ where $\cF$ is the induced foliation. We further assume that the leaf space $\cQ=M/\cF$ is a smooth manifold. Let $q:M\longrightarrow\cQ$ be the quotient map, which is covered by its derivative $\de q: TM \longrightarrow T\cQ$. In the splitting $TM=L_+\oplus L_-$ induced by $K,$ the vector bundle morphism $\de q$ is fibrewise bijective when restricted to $L_+,$ i.e. $\de q \rvert_{L_+}: L_+ \longrightarrow T\cQ$ is a fibrewise isomorphism. Hence the $C^\infty(M)$-module $\sfGamma(L_+)$ is isomorphic to the $C^\infty(\cQ)$-module $\sfGamma(T\cQ).$

We assume that the Riemannian metric $\cH$ is \emph{bundle-like} with respect to the foliation $\cF$, that is,
\begin{align*}
\pounds_{\sfp_-(X)}\cH\big(\sfp_+(Y),\sfp_+(Z)\big) = 0 \ ,
\end{align*}
for all $X,Y,Z\in \sfGamma(TM)$. Then $(M,\cH,\cF)$ is a Riemannian foliation, and the leaf space $\cQ$ admits a Riemannian metric $g$ such that the quotient map $q: M \longrightarrow \cQ$ is a Riemannian submersion. We further assume that the fundamental 2-form $\omega$ is transversally invariant with respect to the foliation $\cF$, that is,
\begin{align*}
\pounds_{\sfp_-(X)}\omega\big(\sfp_+(Y),\sfp_+(Z)\big) = 0 \ ,
\end{align*}
for all $X,Y,Z\in\sfGamma(TM)$. Then $\cQ$ admits a 2-form $b \in {\mathsf\Omega}^2(\cQ)$ inherited from~$\omega.$ 

In this way, the leaf space becomes the target space of a non-linear sigma-model $S(g,b)$ whose background is the $d$-dimensional Riemannian manifold $(\cQ, g)$ with Kalb-Ramond field $b \in {\mathsf\Omega}^2(\cQ)$. The action functional of this sigma-model is
$$
S[\phi]=\frac{1}{2}\,\int_\Sigma\, \bar g_{ij}\, \de \phi^i \wedge
\star\,  \de \phi^j +  \int_\Sigma\, \bar b\ ,
$$
where here $\phi$ is a map from $(\Sigma,h)$ to $(\cQ,g).$

\begin{remark}
For an admissible almost para-Hermitian manifold $(M,K,\eta)$, it is also possible to reduce the corresponding Wess-Zumino action functional $\cS_{\textrm{\tiny WZ}}$ to the leaf space if the canonical $3$-form $H_{\tt can}$ on $M$ additionally satisfies
\begin{align*}
\iota_{\sfp_-(X)} H_{\tt can} = 0 \ , 
\end{align*}
for all $X\in\sfGamma(TM)$ (see~\cite{Severa2019}). Then $H_{\tt can}$ can be regarded as the pullback of a $3$-form on the leaf space $\cQ$ by the quotient map $q:M\longrightarrow\cQ$.
\end{remark}

Suppose now that $L_{\textrm{\tiny D}}\subset TM$ is a Born D-brane such that $W_{\textrm{\tiny D}}=L_+\cap L_{\textrm{\tiny D}}$ has constant rank (cf.~Proposition~\ref{prop:LDWD}). This gives another foliation $\cF_{\textrm{\tiny D}}$ of the Born manifold $M$ whose leaves can be understood as supported by the physical D-branes (cf.~Remark~\ref{rem:dbranesupport}). Then the corresponding D-brane for the sigma-model into $\cQ$ is given by the image of $L_{\textrm{\tiny D}}$ under the derivative $\de q$, $\de q(L_{\textrm{\tiny D}}) =\de q(W_\dee)  \subseteq T\cQ.$ This is a vector sub-bundle of $T\cQ$.
In particular, $\de q(L_{\textrm{\tiny D}})$ is involutive because $L_{\textrm{\tiny D}}$ is involutive, and the restriction of the metric $g$ to $\de q(L_{\textrm{\tiny D}})$ is positive-definite because $q:(M,\cH)\longrightarrow (\cQ,g)$ is a Riemannian submersion. The integrability condition inherited from $L_{\textrm{\tiny D}}$ implies that $\cQ$ admits a regular foliation $\cF^q_{\textrm{\tiny D}}$ whose leaves are supported by the physical D-branes of the sigma-model $S(g,b)$, where the dimension of each leaf is bounded from above by the rank of $W_{\textrm{\tiny D}}=L_+\cap L_{\textrm{\tiny D}}.$

\begin{example}
Let $(M, K, \eta, \cH)$ be a Born manifold which admits a Riemannian foliation $\cF$ such that $L_-=T\cF$ as above. Then $L_-$ is a Born D-brane and its induced physical D-branes are just points ($0$-branes) in $\cQ=M/\cF$. This corresponds to fully Dirichlet boundary conditions for the sigma-model $S(g,b)$ into $\cQ$.

At the opposite extreme, if the eigenbundle $L_+$ of $K$ is integrable as well, then $L_+$ is also a Born D-brane which induces a space-filling physical D-brane whose support is simply the whole leaf space $\cQ$ of the foliation. This corresponds to fully Neumann boundary conditions for the sigma-model $S(g,b)$ into $\cQ$.
\end{example}

\section{Spacetime Perspective: D-Branes on Metric Algebroids} 
\label{sec:targetbranes}

The definition and properties of D-branes given in \cite{Albertsson2009,Sakatani:2020umt} and in Section~\ref{BSM} of the present paper, although physically well motivated from the perspective of open string sigma-models, highlight the construction of a brane solely from the point of view of sub-bundles of the tangent bundle $TM$ of the doubled spacetime. On the one hand, this leads to a more general definition of branes which does not involve submanifolds of $M$. On the other hand, to recover the usual geometric picture of D-branes with worldvolumes and Chan-Paton bundles requires, among other things, the imposition of integrability of the sub-bundle as an extra condition, which cannot be derived from the analysis of the worldsheet constraints alone. As discussed in Remarks~\ref{rem:dbranesupport} and~\ref{rem:B+F}, from this perspective D-branes only arise as foliations, i.e. a single D-brane is given by a submanifold composing the foliation which integrates the distribution that solves the constraints.

The purpose of this section is to develop a complimentary picture of D-branes entirely from the perspective of para-Hermitian geometries on the spacetime, and to discuss how it connects to our worldsheet perspective from Section~\ref{BSM} for integrable branes. In particular, we provide classes of D-branes for Born sigma-models which serve as physically motivated examples of the branes in this section. Regardless of these connections, the treatment which follows is much more general and encompassing, and it mimics the well-known treatment of branes in generalised geometry. The material of this section relies heavily on the theory of metric algebroids, in the settings presented in Appendices~\ref{sec:malg} and~\ref{sec:expreCourant}.

\medskip

\subsection{D-Structures and Branes}\label{sec:Dstructures} ~\\[5pt]
The natural notion of integrability on a metric algebroid is provided by a `D-structure', which was given in~\cite{Freidel2019}. We start by introducing a different notion.

\begin{definition}\label{def:Dstructure}
An \emph{almost D-structure} on a metric algebroid $(E, \eta, \rho,\llbracket\,\cdot \, , \, \cdot\,\rrbracket)$ is an isotropic vector sub-bundle $L \subset E.$  A \emph{D-structure} is an almost D-structure $L$ which is involutive with respect to the D-bracket $\llbracket\,\cdot \, , \, \cdot\,\rrbracket,$ i.e.
$\llbracket\sfGamma(L), \sfGamma(L)\rrbracket \subseteq \sfGamma(L),$ and in this case we say that $L$ is \emph{D-integrable}.

If $(E, \eta, \rho,\llbracket\,\cdot \, , \, \cdot\,\rrbracket)$ is an exact Courant algebroid, then a D-structure is called a \emph{small Dirac structure}. A \emph{Dirac structure} is a small Dirac structure which is maximally isotropic.
\end{definition}

In this definition we allow for sub-bundles with non-constant rank. An almost D-structure $L$ is said to be \emph{regular} if its rank is constant. In particular, if $L$ is a Dirac structure, then the restriction of the Dorfman bracket to $L$ is skew-symmetric, and thus a Dirac structure is a Lie algebroid.

\begin{example} \label{ex:Dstructure}
Let $(E, \eta, \rho,\llbracket\,\cdot \, , \, \cdot\,\rrbracket)$ be an exact pre-Courant algebroid over a manifold $M$ with a maximally isotropic splitting $\sigma$, so that $(E, \eta, \rho,\llbracket\,\cdot \, , \, \cdot\,\rrbracket)$ is isomorphic to the pre-Courant algebroid $(\IT M,\eta_{\IT M},{\rm pr}_{TM},\llbracket\,\cdot\,,\,\cdot\,\rrbracket_{H_\sigma})$, where $H_\sigma\in{\sf\Omega}^3(M)$ is given by \eqref{eq:Hsigma}. Then ${\rm im}(\rho^*)\simeq T^*M$ is involutive with respect to $\llbracket\,\cdot \, , \, \cdot\,\rrbracket$, hence it is a D-structure. 

On the other hand, the $3$-form $H_\sigma$ measures the violation of involutivity of ${\rm im}(\sigma)\simeq TM$ with respect to the D-bracket, and thus ${\rm im}(\sigma)$ is an almost D-structure which is D-integrable if and only if~$H_\sigma=0$, in which case it is a Dirac structure.
\end{example}

We can establish a sufficient condition for a maximally isotropic integrable distribution on an almost para-Hermitian manifold $(M,K,\eta)$ to be integrable with the respect to the D-bracket of the canonical metric algebroid on its tangent bundle $TM$. This is contained in

\begin{proposition}\label{prop:DintDbrane}
Let $(M, K, \eta)$ be an almost para-Hermitian manifold and $L_{\textrm{\tiny D}} \subset TM$ an almost Born D-brane which is preserved by the Levi-Civita connection $\nabla^{\textrm{\tiny\tt LC}}$ of $\eta.$ Then $L_{\textrm{\tiny D}}$ is a D-integrable Born D-brane.
\end{proposition}
\begin{proof}
Since the distribution $L_{\textrm{\tiny D}}$ is preserved by a torsion-free connection, it is Frobenius integrable. Hence $L_{\textrm{\tiny D}}$ is a Born D-brane.

To prove D-integrability we show that 
\be\nonumber
\eta([\![X,Y]\!]_{\tt can},Z)=0 \ ,
\ee
for all $X, Y, Z \in \sfGamma(L_{\textrm{\tiny D}}).$
For this, recall from Example~\ref{ex:canmetricalg} that the canonical connection $\nabla^{\tt can}$ is given by
\be\nonumber
\eta(\nabla^{\tt can}_X \, Y, Z)= \eta(\nabla^{\textrm{\tiny\tt LC}}_X \, Y, Z) -\tfrac12\, \nabla^{\textrm{\tiny\tt LC}}_X\, \omega(Y, K(Z)) \ ,
\ee
where here we restrict to $X, Y, Z \in \sfGamma(L_{\textrm{\tiny D}}).$
Locally we have
\be\nonumber
\nabla^{\textrm{\tiny\tt LC}}\, \omega = \de \omega + \varGamma^{\textrm{\tiny\tt LC}} \wedge \omega \ , 
\ee
where $\varGamma^{\textrm{\tiny\tt LC}}$ is the (local) connection  1-form characterising $\nabla^{\textrm{\tiny\tt LC}}.$ Then
\be \label{leviomega}
\nabla^{\textrm{\tiny\tt LC}}_X\, \omega(Y,Z) = \iota_Z \iota_Y \iota_X(\de \omega + \varGamma^{\textrm{\tiny\tt LC}} \wedge \omega) \ ,
\ee
for all $X, Y, Z \in \sfGamma(L_{\textrm{\tiny D}}).$ 
Since $L_{\textrm{\tiny D}}$ is involutive and isotropic, it follows that 
\be\nonumber
\iota_Z \iota_Y \iota_X \de \omega= - \iota_Z \iota_{[X,Y]} \omega =0 \ ,
\ee
and similarly the second term on the right-hand side of \eqref{leviomega} vanishes.
Thus $\nabla^{\textrm{\tiny\tt LC}}_X\, \omega(Y,Z) = 0$ and 
\be\nonumber
\eta(\nabla^{\tt can}_X \, Y, Z)= \eta(\nabla^{\textrm{\tiny\tt LC}}_X \, Y, Z)
\ee
for all $X, Y, Z \in \sfGamma(L_{\textrm{\tiny D}}),$ since $K(Z) \in \sfGamma(L_{\textrm{\tiny D}}).$
Therefore 
\be\nonumber
\eta([\![X,Y]\!]_{\tt can},Z)= \eta(\nabla^{\textrm{\tiny\tt LC}}_X \, Y-\nabla^{\textrm{\tiny\tt LC}}_Y \, X, Z) + \eta(\nabla^{\textrm{\tiny\tt LC}}_Z \, X, Y) \ ,
\ee
which vanishes because $\nabla^{\textrm{\tiny\tt LC}}$ preserves $L_{\textrm{\tiny D}}.$
\end{proof}

Proposition~\ref{prop:DintDbrane} describes properties of D-branes on the canonical metric algebroid of an almost para-Hermitian manifold, and it inspires the following notion.
Let $(E, \eta, \rho, \llbracket\, \cdot \, , \, \cdot\, \rrbracket)$ be a metric algebroid endowed with a para-complex structure $K \in {\sf Aut}_{\mathds{1}}(E).$ The quintuple  $(E, K, \eta,\rho, \llbracket\, \cdot \, , \, \cdot\, \rrbracket)$ is called a \emph{split metric algebroid}.
Examples are given by split exact pre-Courant algebroids, see Appendix~\ref{sec:expreCourant}. We can provide a natural notion of a brane on a split metric algebroid, which may be viewed as the real counterpart of the notion of a generalised complex brane from~\cite{gualtieri:tesi}.

\begin{definition} \label{brane}
A \emph{brane on a split metric algebroid} $(E, K, \eta, \rho, \llbracket\,\cdot \, , \, \cdot\,\rrbracket)$ is a D-structure $L\subset E$ which is preserved by $K,$ i.e.~$K(L)= L.$
\end{definition}

\begin{example}\label{ex:branepreC}
Let  $(E, K_\sigma, \eta, \rho, \llbracket\,\cdot \, , \, \cdot\,\rrbracket)$ be a split exact pre-Courant algebroid over a manifold $M$. Then $T^*M \subset E$ is a brane.
\end{example}

Beyond Example~\ref{ex:branepreC}, Definition \ref{brane} as it stands is too general to lead to any meaningful insight into the properties of such branes. Moreover, they do not immediately offer a relation to the more physically intuitive geometric structures surrounding D-branes. Hence we proceed to develop a theory of branes on the slightly stronger structure of a pre-Courant algebroid. One advantage provided by this restricted class of metric algebroids is that the notion of D-integrability leads to Frobenius integrability: If $L\subset E$ is a D-structure, then the bracket morphism property of the anchor (see Definition~\ref{malg}) implies that its image $\rho(L)\subset TM$ is involutive, and hence induces a foliation of $M$.
Pre-Courant algebroids constitute a physically meaningful intermediary step between the Courant algebroids of generalised geometry, wherein the section constraint is imposed and solved, and the more general metric algebroids of a fully unconstrained doubled geometry; see~\cite{Jonke2018,Chatzistavrakidis:2019huz,Mori:2020yih,Marotta:2021sia} for detailed descriptions of the chain of metric algebroids involved between type~II supergravity and double field theory.

\medskip

\subsection{Generalised Para-Complex D-Branes}\label{sec:GenDbranes} ~\\[5pt]
\label{sec:genparDbrane}
We shall now extend the definition of generalised submanifolds, guided by the analogue constructions in generalised complex geometry from~\cite{Bursztyn2007, Zambon2008}, to the setting of generalised para-complex structures and pre-Courant algebroids. 
We begin with the natural extension of~\cite[Definition~7.1]{Zambon2008}, which is stated for an exact Courant algebroid.  

\begin{definition} \label{def:gensubmanifold}
Let $(E, \eta , \rho, \llbracket \, \cdot \, , \, \cdot \, \rrbracket)$ be an exact pre-Courant algebroid over a manifold $M$ together with a maximally isotropic splitting $\sigma$ such that $(E, \eta , \rho, \llbracket \, \cdot \, , \, \cdot \, \rrbracket)$ is isomorphic to the pre-Courant algebroid
$(\IT M , \eta_{\IT M} , \mathrm{pr}_{TM}, \llbracket \, \cdot \, , \, \cdot \, \rrbracket_{H_\sigma}),$ where $H_\sigma \in \mathsf{\Omega}^3(M).$ A \emph{generalised submanifold} is a pair $(\cW, L),$ where $i \colon \cW \lhook\joinrel\longrightarrow M$ is a submanifold of $M$ such that $i^* \, \de H_\sigma=0$  and $L \subset E$ is a maximally isotropic D-integrable sub-bundle over $\cW$ such that $\rho(L) = T \cW.$ 
\end{definition}

\begin{remark} \label{rmk:gensub}
Let us unravel and discuss the physical significance of this definition. The generalised tangent bundle $L$ of $\cW$ is isomorphic to the sub-bundle $L_\sigma \subset \IT M$ over $\cW$ by the pre-Courant algebroid isomorphism induced by $\sigma,$  where $L_\sigma$ is the graph
\be\nonumber
L_\sigma = \big\{ X + \alpha \in T\cW \oplus T^*M\rvert_\cW \  \big| \ \alpha \rvert_\cW = \iota_X F_\sigma \big\}
\ee
of some $2$-form $F_\sigma \in \mathsf{\Omega}^2(\cW).$ It is easy to show that $L_\sigma$ takes this form by using the fact that it is maximally isotropic, together with $\rho(L) = T\cW$, and applying the results of Section \ref{sec:Lagrangian}. In particular, there is a one-to-one correspondence between maximally isotropic sub-bundles $L$ over $\cW$ such that $\rho(L)=T\cW$ and $2$-forms $F \in \mathsf{\Omega}^2(\cW).$

Then the D-integrability condition for $L$ can be written as
\be\nonumber
0 = \eta_{\IT M}(\llbracket X_1 + \alpha_1, X_2 + \alpha_2 \rrbracket_{H_\sigma}, X_3 + \alpha_3 )= \iota_{X_3} \iota_{X_2} \iota_{X_1} (i^* H_\sigma + \de F_\sigma) \ ,
\ee
for all $X_1 + \alpha_1, \, X_2 + \alpha_2, \, X_3 + \alpha_3 \in \mathsf{\Gamma}(L_\sigma).$
Hence the pair $(\cW,L)$ must satisfy 
\begin{align*}
i^*H_\sigma = - \de F_\sigma \ ,
\end{align*}
which clearly requires $\de (i^* H_\sigma) = 0 .$ This generalises the condition \eqref{eq:HBcan} on the Lagrangian Born D-branes of Section~\ref{sec:LagrangianBorn}.
\end{remark}

Let $(E,\eta,\rho,\llbracket\,\cdot\,,\,\cdot\,\rrbracket)$ be an exact pre-Courant algebroid. In Appendix~\ref{sec:expreCourant} we describe the natural para-Hermitian structures $(K_\sigma,\eta)$ on $E$. More generally, a para-complex structure $\ccK\in{\sf Aut}_\unit(E)$ which is compatible with the metric $\eta$, in the sense of Definition~\ref{parahermvector}, is the analogue in generalised geometry of an almost para-Hermitian structure on a manifold and is called an \emph{almost generalised para-complex structure}~\cite{Hu:2019zro}.

Motivated by the properties of the D-branes for a Born sigma-model (Definition~\ref{dbranepara}), we can now provide a simple extension of \cite[Definition 7.3]{Zambon2008} to almost generalised para-complex structures. 

\begin{definition} \label{def:genbranes}
Let $(E, \eta , \rho, \llbracket \, \cdot \, , \, \cdot \, \rrbracket)$ be an exact pre-Courant algebroid over $M$ together with a maximally isotropic splitting $\sigma$ and an almost generalised para-complex structure $\ccK.$ 
A \emph{generalised para-complex D-brane} supported on $\cW\subseteq M$ is a generalised submanifold $(\cW,L)$ such that $\ccK(L) = L .$
\end{definition}

\begin{example}
Let $\ccK=K_{\sigma}$ be the natural almost generalised para-complex structure from Appendix~\ref{sec:expreCourant}. By Example~\ref{ex:Dstructure}, $(M,TM)$ is a generalised submanifold if and only if $H_\sigma=0$. Then $(M,TM)$ is a space-filling generalised para-complex D-brane, and in this case $F_\sigma=0$ in the correspondence of Remark~\ref{rmk:gensub}.
\end{example}

Definition~\ref{def:genbranes} should be regarded as a \emph{localised} version of our previous notions of D-branes, which is suitable for describing a single worldvolume submanifold instead of a whole foliation. Repeating the analysis of Sections~\ref{sec:Lagrangian} and~\ref{sec:genBornDbranes} shows that the general structure of a generalised para-complex D-brane is analogous to that of an almost Born D-brane from Proposition~\ref{prop:LDWD}: Let $E=\ccL_+\oplus\ccL_-$ be the decomposition of $E$ into the $\pm\,1$-eigenbundles $\ccL_\pm$ of $\ccK$. Suppose that $\ccW:=L\cap \ccL_+$ has constant rank, and set $\overline{\ccW}:=\eta^\sharp\big({\sf Ann}(\ccW)\big)$. Then
\begin{align*}
L = \ccW \, \oplus \, \overline{\ccW} \ .
\end{align*}
When $i^*H_\sigma=0$, suitable integral $2$-forms $F_\sigma$ from the correspondence of Remark~\ref{rmk:gensub} yield gauge fluxes on the worldvolume $\cW$. We illustrate these structures in the examples of Born D-branes from Section~\ref{BSM}.

\begin{example}\label{ex:BornDbranes}
We show how the Born D-branes of Section~\ref{sec:Dbranes} provide special instances of the generalised para-complex D-branes of Definition \ref{def:genbranes}. Let $(M,K, \eta)$ be an almost para-Hermitian manifold. The standard Courant algebroid $(\IT M,\eta_{\IT M},\mathrm{pr}_{TM},\llbracket \, \cdot \, , \, \cdot \, \rrbracket_0)$ over $M$ is called the \emph{large Courant algebroid}, and it features in the construction of the metric algebroids of double field theory~\cite{Jonke2018,Chatzistavrakidis:2019huz,Marotta:2021sia,Hu:2019zro}. It admits an almost generalised para-complex structure
\be \label{eq:genparaK}
\ccK_K =
\begin{pmatrix}
K & 0 \\
0 & -K^{\mathtt{t}}
\end{pmatrix}
\ee  
that clearly preserves the splitting $\IT M = TM \oplus T^*M.$ 

Consider a Born D-brane in the almost para-Hermitian manifold $M$, i.e. a maximally isotropic integrable distribution $L_{\textrm{\tiny D}}$ such that $K(L_{\textrm{\tiny D}})=L_{\textrm{\tiny D}}.$ Let $i:\cW_{\textrm{\tiny D}}\lhook\joinrel\longrightarrow M$ be an embedded leaf of the foliation $\cF_{\textrm{\tiny D}}$ induced by $L_{\textrm{\tiny D}}$. Then $T\cW_\dee \oplus {\mathsf{Ann}}(T\cW_\dee) \subset \IT M$
forms a generalised submanifold with $\cW_\dee$, since this sub-bundle is clearly maximally isotropic and D-integrable with respect to the standard Dorfman bracket $\llbracket \, \cdot \, , \, \cdot \, \rrbracket_0$. Since 
\be \nonumber
\ccK_K\big(T\cW_\dee \oplus {\mathsf{Ann}}(T\cW_\dee)\big)= T\cW_\dee \oplus {\mathsf{Ann}}(T\cW_\dee) \ ,
\ee
it follows that $\big(\cW_\dee, T\cW_\dee \oplus {\mathsf{Ann}}(T\cW_\dee)\big)$ is a generalised para-complex D-brane on $M$. This is analogous to the splitting property of Proposition~\ref{prop:LDWD}, and in particular the sub-bundle $T\cW_\dee \oplus {\mathsf{Ann}}(T\cW_\dee)$ is also invariant with respect to the natural para-complex structure of $\IT M$ given by (see Example~\ref{ex:gentanbun})
\be\nonumber
K_{\IT M}=
\begin{pmatrix}
\unit & 0 \\
0 & -\unit
\end{pmatrix}
\ ,
\ee  
as discussed in Section \ref{sec:Lagrangian}.

More generally, if $(M,K,\eta)$ is admissible with closed canonical $3$-form $H_{\tt can}$, we can consider its corresponding $H_{\tt can}$-twisted large Courant algebroid $(\IT M,\eta_{\IT M},\mathrm{pr}_{TM},\llbracket \, \cdot \, , \, \cdot \, \rrbracket_{H_{\tt can}})$. Then the pair $\big(\cW_\dee, T\cW_\dee \oplus {\mathsf{Ann}}(T\cW_\dee)\big)$ is a generalised para-complex D-brane if and only if $i^*H_{\tt can}=0.$ In other words, these D-branes cannot support non-zero generalised fluxes.
\end{example}

\begin{example}
Example~\ref{ex:BornDbranes} can be extended in a way which elucidates further the relation between the Lagrangian Born D-branes from Section~\ref{sec:LagrangianBorn} and generalised para-complex D-branes. Let $(M, K, \eta)$ be an admissible almost para-Hermitian manifold together with a Born D-brane $L_\dee$ inducing a foliation $\cF_\dee$ of $M$, and consider the $H_{\tt can}$-twisted large Courant algebroid $(\IT M, \eta_{\IT M}, \mathrm{pr}_{TM}, \llbracket \, \cdot \, , \, \cdot \, \rrbracket_{H_{\tt can}})$ on $M$.  For any embedded leaf $i \colon \cW_\dee \lhook\joinrel\longrightarrow M$ of $\cF_\dee$, we construct a generalised submanifold $\big(\cW_\dee, L^F\big)$ by picking a $2$-form $F \in \mathsf{\Omega}^2(\cW_\dee)$ and applying the bijective correspondence discussed in Remark~\ref{rmk:gensub}, i.e. we set
\be\nonumber
L^F = \big\{ X+ \alpha \in T\cW_\dee \oplus T^*M\rvert_{\cW_\dee}  \ \big| \ \alpha\rvert_{\cW_\dee} = \iota_X F \big\} \ .
\ee
The $2$-form $F$ may be induced by a $B_+$-transformation of the para-Hermitian vector bundle $(TM,K,\eta)$ (as discussed in Remark~\ref{rem:B+F}) or alternatively of $(\IT M,K_{\IT M},\eta_{\IT M})$ (where it acts by pulling back $L^F$ to $L^{F+i^*B}$ for a $2$-form $B\in\mathsf{\Omega}^2(M)$). In any case, in order for $\big(\cW_\dee, L^F\big)$ to define a generalised submanifold, $F$ must satisfy the $H_{\tt can}$-twisted integrability condition 
\begin{align} \label{eq:dFHcan}
\de F + i^*H_{\tt can} = 0 \ .
\end{align}

Let us choose again the almost generalised para-complex structure \eqref{eq:genparaK} on $\IT M$ induced by the almost para-complex structure $K$ on $TM$. If $F=0$, then $L^0 = T\cW_\dee \oplus {\mathsf{Ann}}(T\cW_\dee)$ and we recover the D-branes of Example~\ref{ex:BornDbranes}. This is analogous to the Lagrangian Born D-branes carrying flat connections. 

More generally,
the pair $\big(\cW_\dee, L^F\big)$ is a generalised para-complex D-brane if and only if
\be \label{eq:para-hol}
K^{\mathtt{t}} (\iota_X F) + \iota_{K(X)} F \ \in \ \mathsf{Ann}(T\cW_\dee) \ ,
\ee
for all $X\in T\cW_\dee$,
because the condition $K(T\cW_\dee)=T\cW_\dee$ is ensured by the fact that $L_\dee$ is a Born D-brane.
If $i^*H_{\tt can}=0$ and $F$ is an integral $2$-form, then there exists a complex line bundle $C$ over $\cW_\dee$ with a connection $\nabla^C$ such that $F$ is the curvature of $\nabla^C$. For an integrable almost para-complex structure $K$, by~\cite[Proposition~2]{Lawn2005} the condition \eqref{eq:para-hol} implies that $\big(C,\nabla^C\big)$ defines a \emph{para-holomorphic} line bundle. This is the analogue of B-branes on a complex manifold~\cite{Witten:1992fb}, which come with holomorphic Chan-Paton bundles, and their realisation as generalised complex branes~\cite{gualtieri:tesi}.

The general $H_{\tt can}$-twisted integrability condition \eqref{eq:dFHcan} can be interpreted as saying that the canonical $3$-form $H_{\tt can}$ sources a distribution of magnetic charge on the D-brane. It implies that the geometric description of the `gauge field' on those D-branes which support non-zero generalised fluxes is not simply through a connection on a vector bundle; this is somewhat analogous to the obstructions discussed in~\cite{Freed:1999vc,Kapustin:1999di}. For D-branes in doubled twisted tori, the gauge field is a connection on a module over a bundle of noncommutative algebras~\cite{Hull:2019iuy,Aschieri:2020uqp}.
\end{example}

\medskip

\subsection{Reduction of Large Courant Algebroids and D-Branes} ~\\[5pt]
We will now explain how to implement the reductions of our D-branes to `physical' spacetimes, in the sense of Section~\ref{sec:Dbranesleaf}, within the framework of the present section. This relies heavily on the theory of Courant algebroid reduction developed by~\cite{Bursztyn2007, Zambon2008}, which we review in Appendix~\ref{app:Courantred}.

\begin{example} \label{rmk:standardreduction}
We apply the reduction of Theorem~\ref{thm:BCGZ} to the setting of Section~\ref{sec:Dbranesleaf}. 
Let $(M,K,\eta)$ be an almost para-Hermitian manifold, and assume that the eigenbundle $L_-=T\cF$ of $K$ is integrable.  Thus $M$ is foliated by $\cF$, and we suppose  that the leaf space $\cQ= M/\cF$ is a smooth manifold. Then there is a unique surjective submersion $q: M \longrightarrow \cQ$ which is compatible with the smooth structure of $\cQ.$ 

Consider the large Courant algebroid $(\IT M, \eta_{\IT M} , \mathrm{pr}_{TM}, \llbracket \, \cdot \, , \, \cdot \, \rrbracket_0)$ on $M$. We  set $$A \coloneqq  L_- \oplus \{0\} = T\cF \oplus \{ 0 \} \ .$$ Then $A^\perp$ is spanned fibrewise by  sections of the form $Y + \de (q^*f),$ where $Y \in \mathsf{\Gamma}(TM)$ is a projectable vector field\footnote{A vector field $Y\in\sfGamma(TM)$ is \emph{projectable} with respect to the foliation $\cF$ if $[X,Y]\in\sfGamma(T\cF)$ for all $X\in\sfGamma(T\cF)$.} and $f \in C^\infty(\cQ).$ To apply Theorem~\ref{thm:BCGZ} to this case, we need to check that $Y + \de (q^*f)$ is basic with respect to $A$ (Definition~\ref{def:basic}). For this, we compute
\be \nonumber
\llbracket X , Y + \de (q^*f) \rrbracket_0 = [X,Y] + \pounds_X\, \de (q^*f) =  [X,Y] \ \in \  \mathsf{\Gamma}(A)  \ ,
\ee
for any $X \in \mathsf{\Gamma}(A)$, and therefore $Y + \de (q^*f) \in \mathsf{\Gamma}_{\mathtt{bas}}(A^\perp).$ 
Theorem \ref{thm:BCGZ} then implies that the reduced Courant algebroid is given by the standard Courant algebroid on $\cQ,$ through the pullback diagram
\be \nonumber
\begin{tikzcd}
\displaystyle \frac{(L_- \oplus \{ 0 \})^\perp}{L_- \oplus \{ 0 \}} \arrow{r}{} \arrow{dd} & \IT\cQ \arrow{dd}{} \\ & \\
M \arrow{r}{q} & \cQ
\end{tikzcd}
\ee
This makes rigorous previous arguments suggesting that the large Courant algebroid on a doubled space should reduce to the standard Courant algebroid on a `physical spacetime' upon implementation of the section constraint, see e.g~\cite{Jonke2018,Chatzistavrakidis:2019huz,Marotta:2021sia}.
\end{example}

\begin{example} \label{ex:redDirac}
Building on Example \ref{rmk:standardreduction}, we apply the Dirac reduction of Proposition~\ref{prop:Dircred} to a Born D-brane. For this, we assume that the fundamental $2$-form $\omega$ of the almost para-Hermitian manifold $(M,K,\eta)$ is transversally invariant and introduce a bundle-like Born metric $\cH$ with respect to the foliation $\cF$ induced by the involutive sub-bundle $L_-\subset TM$. A Born D-brane $L_\dee \subset TM$ for the Born sigma-model $\cS(\cH,\omega)$ into $M$ induces another foliation $\cF_\dee$ of $M$, as well as a Dirac structure for the large Courant algebroid on $M$ given by
\be \nonumber
L := L_\dee \oplus {\mathsf{Ann}}(L_\dee) \ .
\ee

With the reduction discussed in Example \ref{rmk:standardreduction}, we need to ensure that the hypotheses of Proposition~\ref{prop:Dircred} are met, which restricts the class of Born D-branes that can be reduced to the leaf space in this way. 
In particular, the condition that $L \cap A^\perp$ has constant rank is satisfied if $L_\dee = T\cF_\dee$ admits a sub-bundle spanned by projectable vector fields. Then the sections which span $L \cap A^\perp$ still take the form $Y + \de (q^*f),$ where $Y \in \mathsf{\Gamma}(L_\dee)$ is a projectable vector field and
$\de (q^*f) \in \mathsf{\Gamma}(\mathsf{Ann}(L_\dee))$ with $f \in C^\infty(\cQ).$ Thus the condition \eqref{eq:diracred1} is satisfied because
\be \nonumber
\llbracket X , Y + \de (q^*f) \rrbracket_0 = [X,Y] \ \in \ \mathsf{\Gamma}(L_-) \ , 
\ee
for all $X \in \mathsf{\Gamma}(A)= \mathsf{\Gamma}(L_- \oplus \{ 0 \}).$ The condition \eqref{eq:diracred2} is satisfied because $L_\dee$ is an integrable sub-bundle. Hence by Proposition~\ref{prop:Dircred}, $L$ descends to a Dirac structure $L_{\mathtt{red}}$ on $\cQ = M/\cF.$ In this way the reduction of a Born D-brane to the leaf space can also be regarded as a Dirac reduction. The interpretation of the set of D-branes for the quotient sigma-model $S(g,b)$ into $\cQ$ as a Dirac structure is analogous to the considerations of D-branes in generalised geometry by~\cite{Asakawa2012}.
\end{example}

\begin{example}
Finally we consider the reduction of the generalised para-complex D-branes which were studied in Section~\ref{sec:genparDbrane}.
The reduction of a generalised complex structure discussed in~\cite[Proposition~6.1]{Zambon2008} adapts to the case of a generalised para-complex structure. This restricts the class of generalised para-complex D-branes that can be reduced to the leaf space in this way. In particular, D-integrability of a generalised para-complex structure $\ccK$ on an exact Courant algebroid $(E, \eta, \rho, \llbracket\,\cdot \, , \, \cdot\,\rrbracket)$ requires that its Nijenhuis tensor vanishes:
\be \nonumber
\mathsf{Nij}_{\ccK}(e_1, e_2) \coloneqq \llbracket \ccK(e_1), \ccK(e_2) \rrbracket - \llbracket e_1, e_2 \rrbracket - \ccK\big(\llbracket e_1, \ccK(e_2) \rrbracket +\llbracket \ccK(e_1), e_2\rrbracket\big) = 0 \ ,
\ee
for all $e_1,e_2\in\sfGamma(E)$. Let $A \longrightarrow \cW$ be a sub-bundle of $E$ over a submanifold $\cW$ satisfying the conditions of Theorem \ref{thm:BCGZ}.
Then the hypotheses of \cite[Proposition~6.1]{Zambon2008} further require that $\ccK (A) \cap A^\perp$ has constant rank, that
\be \nonumber
\ccK (A) \cap A^\perp \ \subseteq \ A \ , 
\ee
and that
\be \nonumber
\ccK\big(\mathsf{\Gamma}_{\mathtt{bas}}\big(\ccK (A) \cap A^\perp\big)\big) \ \subseteq \  \mathsf{\Gamma}_{\mathtt{bas}}\big(\ccK (A) \cap A^\perp\big) \ .
\ee

For the large Courant algebroid over an almost para-Hermitian manifold $(M, K, \eta)$, the almost generalised para-complex structure $\ccK_K$ induced by $K$ is given by \eqref{eq:genparaK}. Its Nijenhuis tensor $\mathsf{Nij}_{\ccK_K}$ vanishes if and only if the almost para-complex structure $K$ is Frobenius integrable. In general, this is not the case for the almost para-Hermitian structure which induces a Born sigma-model, where only one of the eigenbundles $L_\pm$ of $K$ is required to be integrable; in fact, this happens in most examples of physical interest. Thus the reduction of a generalised para-complex D-brane on $M$ does not necessarily yield a generalised para-complex structure on the leaf space $\cQ=M/\cF$. This harmonises nicely with the expected picture of a D-brane in the `physical spacetime'.
\end{example}

\medskip

\subsection{D-Branes on Doubled Nilmanifolds} ~\\[5pt]
A large class of consistent compactifictions of supergravity are provided by doubled twisted tori~\cite{DallAgata2007}, which are defined as quotients $M=\sfGamma\setminus{\sf G}$ of a doubled Lie group $\sfG$~\cite{Marotta:2019eqc} by a discrete cocompact subgroup $\sfGamma\subset{\sf G}$ acting from the left; these have gauge algebras realised as the isometry algebra of $M$~\cite{Hull2009}. The smooth manifold $M$ is parallelisable, and so has a basis of globally defined left-invariant $1$-forms. The isometry group of $M$ is generated by vector fields dual to these $1$-forms. Polarisations of $M$ give quotients which are physical spaces realised as torus bundles, with a given monodromy encoded geometrically in $M$.

Here we will focus on the example of the doubled nilmanifold $M_\sfH$, which is obtained as a quotient of the cotangent bundle $T^*\sfH = \sfH\ltimes\frh$ of the three-dimensional Heisenberg group $\sfH$ by a discrete cocompact subgroup $\sfGamma_m$ labelled by an integer $m\in\IZ$. The doubled nilmanifold $M_\sfH$ admits two distinguished fibrations: a principal $\IT^3$-bundle over the Heisenberg nilmanifold $\IT_\sfH$ and a fibration over the $3$-torus $\IT^3$ with $\IT^3$ fibres. For further details, including the reductions of the doubled sigma-models in the corresponding polarisations of $M_\sfH$, see \cite{Marotta:2019eqc, Hull2009}. We shall illustrate our formalism by reproducing some of the D-branes found in~\cite{Lawrence:2006ma,Albertsson2009}.

\begin{example}[\bf Nilmanifold] \label{ex:nilmanifold}
Let $M_\sfH \longrightarrow \IT_\sfH$ be the principal $\IT^3$-bundle induced by the quotient of the Drinfel\rq{}d double $T^*\sfH,$ where the nilmanifold $\IT_\sfH$ has degree $m$ when viewed as a circle bundle over a $2$-torus $\IT^2$. This bundle inherits a split signature metric $\eta$ from the bi-invariant split signature metric on $T^*\sfH$, together with a compatible para-complex structure $K.$ Let
$(x,y,z, \tilde{x}, \tilde{y}, \tilde{z})$ be the local coordinates on $M_\sfH$ descending from the coordinates on $T^*\sfH=\sfH\ltimes\frh.$ Then the eigenbundles of $K$ are spanned pointwise  by
\be
 Z_x=\frac\partial{\partial x} \ , \quad Z_y= \frac\partial{\partial y} \qquad \mbox{and} \qquad  Z_z= \frac\partial{\partial z} + m\,x\,\frac\partial{\partial y} \ , \label{zdist}
\ee
for the distribution $L_+$, and 
\be
 \tilde{Z}^x= \frac\partial{\partial\tilde x} \ , \quad \tilde{Z}^y=  \frac\partial{\partial\tilde y} + m\,z \,  \frac\partial{\partial\tilde x} -m\,x\,\frac\partial{\partial\tilde z} \qquad  \mbox{and} \qquad \tilde{Z}^z= \frac\partial{\partial\tilde z} \ , \label{tilddist}
\ee
for the distribution $L_-.$

Then $L_+$ is an integrable distribution which is preserved by the para-complex structure $K,$ i.e. it is a Born D-brane. Following Example \ref{ex:redDirac}, we consider the sub-bundles  $A = L_- \oplus \{ 0 \}$ and $L=L_+ \oplus {\mathsf{Ann}}(L_+)$ of $\IT M_\sfH$. Then the $TM_\sfH$-component of $L \cap A^\perp$ is given by $L_+,$ i.e. $\mathsf{\Gamma}(L_+)$ is spanned by projectable vector fields. This yields a Dirac structure on the standard Courant algebroid over $\IT_\sfH$, which is given by the tangent bundle $T\IT_\sfH$, i.e. it is a D$3$-brane filling $\IT_\sfH$.

Similarly, the involutive distribution $L_\dee$ defined pointwise by 
\be \nonumber
L_\dee \big\rvert_\xi \coloneqq {\mathsf{Span}}(Z_x, \, Z_y , \,   \tilde{Z}^z) \ ,
\ee 
for all $\xi \in M_\sfH,$ is preserved by $K$ and is maximally isotropic with respect to $\eta.$ Hence it is a Born D-brane which induces the Dirac structure $L= L_\dee \oplus {\mathsf{Ann}}(L_\dee)$ on the large Courant algebroid $\IT M_\sfH.$ The Dirac structure $L$ can be reduced to a Dirac structure on the Heisenberg nilmanifold $\IT_\sfH$
because the $TM_\sfH$-component of the intersection $L \cap A^\perp$ is spanned by $\{ Z_x, \, Z_y \},$ where $A = L_- \oplus \{ 0 \}.$ The reduced Dirac structure is associated with a foliation of $\IT_\sfH$ whose leaves have dimension~$2$ and correspond to tori $\IT^2$. They are D2-branes for the reduced sigma-model into $\IT_\sfH$.

D1-branes on $\IT_\sfH$ with the topology of circles are analogously obtained by considering the maximally isotropic involutive distribution $ L'_\dee$ defined pointwise as 
\be\nonumber
 L'_\dee \big\rvert_\xi \coloneqq {\mathsf{Span}} (Z_y, \, \tilde{Z}_x , \,   \tilde{Z}^z) \ ,
\ee
for all $\xi \in M_\sfH,$ which is preserved by $K.$ In this case the $TM_\sfH$-component of the intersection $L \cap A^\perp$ is spanned by $Z_y,$ where $L =  L'_\dee \oplus {\mathsf{Ann}}( L'_\dee)$ and $A=L_- \oplus \{ 0 \}$.
\end{example}

\begin{example}[\bf ${\IT^3}$ with $\boldsymbol H$-flux]
We consider now the torus bundle $M_\sfH \longrightarrow \IT^3$ with split signature metric $\eta$ as above and compatible almost para-complex structure $K^\prime$ determined by the eigenbundles $L_\pm^\prime$ respectively spanned pointwise by 
\be\label{eq:Zprime}
Z'_x = \frac{\partial}{\partial x} \ , \quad Z'_y= \frac{\partial}{\partial y} -m\,x\, \frac{\partial}{\partial \tilde{z}} \qquad \mbox{and} \qquad Z'_z=\frac{\partial}{\partial z} - m\,y\, \frac{\partial}{ \partial \tilde{ x}} + m\,x\, \frac{\partial}{\partial \tilde{y}} \ ,
\ee 
and
\be\nonumber
\tilde{Z}^{\prime\,x} = \frac\partial{\partial\tilde x} \ , \quad \tilde{Z}^{\prime\,y} = \frac\partial{\partial\tilde y} \qquad \mbox{and} \qquad \tilde{Z}^{\prime\,z} =\frac\partial{\partial\tilde z} \ ,
\ee
where $(x, y, z, \tilde{x}, \tilde{y}, \tilde{z})$ are local coordinates on $M_\sfH$ adapted to the fibres at a point $\xi \in M_\sfH.$ In this case the integer $m$ labels the Dixmier-Douady class of a gerbe in ${\sf H}^3(\IT^3;\IZ)\simeq\IZ$.

An example of a Born D-brane is given by the subbundle $L_\dee\subset TM_\sfH$ defined pointwise as
\be\nonumber
L_\dee \big\rvert_\xi \coloneqq {\mathsf{Span}} (Z^\prime_x, \, \tilde{Z}^{\prime\,y} , \,   \tilde{Z}^{\prime\,z} )
\ee
for any $\xi \in M_\sfH.$ This sub-bundle is clearly preserved by $K^\prime$, it is maximally isotropic with respect to $\eta,$ and it is integrable:
\be\nonumber
[Z^\prime_x, \tilde{Z}^{\prime\,y}] = [Z^\prime_x, \tilde{Z}^{\prime\,z}] = [\tilde{Z}^{\prime\,y}, \tilde{Z}^{\prime\,z} ]  = 0 \ .
\ee
The Born D-brane $L_\dee$ can be reduced by considering its Dirac structure $L \coloneqq L_\dee \oplus {\mathsf{Ann}}(L_\dee).$ In this case, we set $A = L^\prime_- \oplus \{ 0 \}$, and then the $TM_\sfH$-component of $A^\perp \cap L$ is given by the sub-bundle spanned pointwise by $Z^\prime_x.$ This yields a Dirac structure on the quotient manifold $\cQ = \IT^3$ associated with a one-dimensional foliation of $\cQ.$ Its leaves are circles and are wrapped by D$1$-branes. These D-branes are T-dual to the D$0$-branes on the Heisenberg nilmanifold $\IT_\sfH$ obtained from reducing the Born D-brane $L_-$ of Example~\ref{ex:nilmanifold}.

It is well-known that the $3$-torus $\IT^3$ with non-zero $H$-flux does not admit any space-filling D$3$-branes; this follows from the non-vanishing Freed-Witten anomaly~\cite{Freed:1999vc} in this case and was also reproduced by the analysis of~\cite{Albertsson2009}. In our geometric approach this follows immediately: The bundle $M_\sfH \longrightarrow \IT^3$ is non-trivial, hence it does not admit any horizontal integrable distribution of rank $3.$ A less conceptual but more calculational way of seeing this is to use the computations of the D-brackets from~\cite{SzMar} to determine the canonical $3$-form $H_{\tt can}$ on $M_\sfH$ in this polarisation. One finds
\begin{align*}
H_{\tt can} = -\tfrac32 \, m\,\de x\wedge \de y\wedge \de z \ .
\end{align*} 
This is non-zero when evaluated on the basis of local vector fields from \eqref{eq:Zprime} for $m\neq0$, and so it forbids D$3$-branes wrapping $\IT^3$ (cf. Sections~\ref{sec:LagrangianBorn} and~\ref{sec:genparDbrane}).
\end{example}

\appendix

\section{Para-Hermitian Geometry and Metric Algebroids} \label{intropara}

In this appendix we summarise the mathematical background that is used throughout the main text of the paper. The reader should note the similarities with analogous structures defined in generalised geometry. For further details, advancements and references, see e.g.~\cite{Vaisman2012,Vaisman2013,Freidel2017,Svoboda2018,Freidel2019, Marotta:2019eqc,Marotta:2021sia}.

\subsection{Para-Hermitian Vector Bundles} \label{sec:paravector} ~\\[5pt]
We start by introducing the key geometric structures that play a key role throughout our discussions.
\begin{definition} \label{parahermvector}
Let $E\longrightarrow  M$ be a real vector bundle of even rank $2d$ over a smooth manifold $ M$. A
\emph{para-complex structure} on $E$ is a vector bundle automorphism
$K \in {\sf Aut}_\unit(E)$ covering the identity such that
$K^2=\unit$ and $K\neq\pm\,\mathds{1}$, and whose $\pm\,1$-eigenbundles  have equal rank $d$. The pair $(E,K)$ is a \emph{para-complex vector bundle}. 

If $E$ additionally admits a fibrewise  metric $\eta\in \mathsf{\Gamma}\big(\midodot^2 E^*\big)$ of split signature $(d,d)$ which is compatible with the para-complex structure $K$ in the sense that\footnote{Throughout we use the same symbol for a vector bundle morphism covering the identity and its induced $C^\infty( M)$-module morphism on sections.}
$$
\eta\big(K(e_1),K(e_2)\big)=-\eta(e_1,e_2) \ , 
$$
for all $e_1, e_2 \in \sfGamma(E)$, then the pair $(K, \eta)$ is
a \emph{para-Hermitian structure} on $E$ and the triple $(E,K,\eta)$ is a \emph{para-Hermitian vector bundle}. 

An \emph{almost para-Hermitian manifold} is a triple $(M,K,\eta)$ of a manifold $M$ of even dimension $2d$ with an almost para-Hermitian structure $(K,\eta)$ on its tangent bundle $TM$.\footnote{We may drop the adjective `almost' if $K$ is Frobenius integrable.}
\end{definition}

In this case $K$ admits two eigenbundles $L_\pm$ with
eigenvalues $\pm\,1,$ so that 
$$
E= L_+ \oplus L_- \ ,
$$ 
and $L_\pm$ are maximally
isotropic with respect to the fibrewise metric $\eta.$ The case $E=TM$ for an
almost para-Hermitian manifold $M$ is particularly relevant because it
allows one to formulate conditions for the integrability of the
eigenbundles $L_\pm\subset TM$, and hence on the possibility that $M$ is a foliated
manifold.

\begin{remark}\label{rem:EL-}
Let $E\longrightarrow M$ be a vector bundle of rank $2d$ endowed with a split
signature metric $\eta$, and let $L$ be a maximally isotropic sub-bundle
of $E.$ Then the short exact sequence
\be
0\longrightarrow L\longrightarrow E\longrightarrow
E/L\longrightarrow 0 
\label{eq:EL-}\ee
always admits a maximally isotropic splitting. This determines a para-Hermitian structure on $E.$ All  maximally isotropic splittings of  \eqref{eq:EL-} give isomorphic para-Hermitian structures.
\end{remark}

It is straightforward to see that the compatibility condition between $\eta$ and $K$ in Definition~\ref{parahermvector} is equivalent to 
$$
\eta \big(K(e_1), e_2\big)= -\eta \big(e_1, K(e_2)\big) \ , 
$$ 
for all $e_1,e_2 \in \sfGamma(E)$.  Any para-Hermitian vector
bundle $E$ is therefore endowed with a non-degenerate \emph{fundamental $2$-form} $\omega\in \mathsf{\Gamma}\big(\midwedge^2 E^*\big)$ given by
$$
\omega(e_1,e_2)= \eta\big(K(e_1), e_2\big) \ , 
$$ 
for all $e_1,\e_2 \in \sfGamma(E)$. The eigenbundles $L_\pm \subset E$ are also maximally isotropic with respect to $\omega.$

\begin{example} \label{ex:gentanbun}
Let $E=\mathbb{T} M$ be the \emph{generalised tangent bundle }
$$
\mathbb{T} M= T M \oplus T^* M
$$ 
over a manifold $ M.$ It is naturally endowed with a fibrewise split signature metric 
$$
\eta_{\IT M}(X+ \alpha, Y+ \beta)= \iota_X \beta + \iota_Y \alpha \ , 
$$
for all $X+\alpha,Y+ \beta \in \mathsf{\Gamma}(\mathbb{T} M),$ where $\iota_X$ is the interior multiplication of forms by the vector field $X\in\sfGamma(T M)$. The natural para-complex structure $K$ of $\mathbb{T} M$ is given by
$$
K_{\IT M}(X+\alpha)= X- \alpha \ ,
$$ 
for all $X+\alpha \in \mathsf{\Gamma}(\mathbb{T} M),$ so that $T M$ and
$T^* M$ are the respective $\pm\,1$-eigenbundles. Clearly $\eta$ and
$K$ are compatible in the sense of Definition~\ref{parahermvector},
and the bundles $T M$ and $T^* M$ are maximally isotropic with respect to $\eta.$ Thus we obtain a fundamental $2$-form
$$
\omega_{\IT M}(X+ \alpha, Y+ \beta)= \iota_X \beta - \iota_Y \alpha \ , 
$$ 
for all $X+\alpha , Y+ \beta \in \mathsf{\Gamma}(\mathbb{T} M),$ which is the additional natural non-degenerate pairing that can be defined in this case \cite{gualtieri:tesi}. 
\end{example}

\medskip

\subsection{Generalised Metrics and Born Geometry} \label{genmesubs} ~\\[5pt]
We shall now introduce a generalised metric compatible with a para-Hermitian structure.

\begin{definition} \label{genme}
Let $(E,\eta)$ be a pseudo-Euclidean vector bundle over a manifold $ M$. A \emph{generalised 
  metric} on $E$ is an automorphism $I \in {\sf Aut}_\unit(E)$ with
$I^2=\unit$ and $I\neq\pm\,\mathds{1}$ which defines a fibrewise Riemannian metric
\be \nonumber
\cH (e_1,e_2)=\eta\big(I(e_1),e_2\big) \ ,
\ee  
for all $e_1,e_2 \in \sfGamma(E)$.
\end{definition}

\begin{remark}
Definition \ref{genme} can also be recast in a different form.
Any generalised metric induces a vector bundle isomorphism\footnote{Throughout this paper a superscript ${}^\sharp$ denotes
the bundle morphism $E^*\longrightarrow E$ induced by a 
$(2,0)$-tensor in $\mathsf{\Gamma}(E\otimes E)$. For a $(0,2)$-tensor in $\mathsf{\Gamma}(E^*\otimes E^*)$ we
 use a superscript ${}^\flat$ for the induced bundle morphism
$E\longrightarrow E^*$. Conversely, the tensor associated
to a vector bundle morphism $T$ will be 
underlined as $\underline{T}\,$.}  $\cH^\flat \in {\sf Hom}(E, E^*)$ which, in terms of the induced vector bundle isomorphisms $\eta^{-1}{}^\sharp, \cH^{-1}{}^\sharp\in {\sf Hom}(E^*, E),$ satisfies the condition
\be \label{condgenme}
\eta^{-1}{}^\sharp\big(\cH^\flat(e)\big)= \cH^{-1}{}^\sharp\big(\eta^\flat(e)\big) \ ,
\ee
for all $e \in E,$ so that
$\eta^{-1}{}^\sharp \circ \cH^\flat \in {\sf End}(E)$. In Definition \ref{genme}
this is nothing but $I= \eta^{-1}{}^\sharp \circ \cH^\flat \in {\sf Aut}_\unit(E),$
and \eqref{condgenme} implies that $\eta^{-1}{}^\sharp\circ \cH^\flat$ squares to the
identity map in ${\sf End}(E).$ The tensor induced by this map can
be regarded as a section $\underline{I} \in \mathsf{\Gamma}(E^* \otimes E).$
\end{remark}

\begin{example} \label{ex:etacHrel}
Let $(M, K, \eta)$ be
an almost para-Hermitian manifold. A generalised metric on  $M$ is defined by
$$
\cH(X,Y):=\eta\big(I(X),Y\big) \ , 
$$
for all $X, Y \in \sfGamma(TM)$,
where  $I \in {\sf Aut}_\unit(TM)$ with $I^2=\unit$ and $I\neq\pm\,\mathds{1}$. It satisfies 
$$
\eta^{-1}{}^\sharp\big(\cH^\flat(X)\big)= \cH^{-1}{}^\sharp\big(\eta^\flat(X)\big) \ ,
$$
for all $X \in TM.$ Then $I(X)= \eta^{-1}{}^\sharp(\cH^\flat(X)),$
for all $X \in TM$.
\end{example}

We can unravel the structure of a generalised metric through~\cite{Marotta:2019eqc}

\begin{proposition} \label{gbparaherm}
Let $(E,K, \eta)$ be a para-Hermitian vector bundle. A generalised metric $\cH$ on $E$ defines a unique pair $(g_+, b_+)$
of a fibrewise Riemannian metric $g_+ \in \mathsf{\Gamma}(\midodot^2 L_+^*)$ on the
sub-bundle $L_+\subset E$ and a 2-form $b_+\in \mathsf{\Gamma}(\midwedge^2 L_+^*).$
Conversely, any such pair $(g_+,b_+)$ uniquely defines a generalised metric.
\end{proposition}

\begin{example}
Let $E=\mathbb{T} M=T M\oplus
T^* M$ be the generalised tangent bundle over a manifold $ M$. A generalised metric $\cH$ on  $ \mathbb{T} M$ is equivalent to
a Riemannian metric $g_+$ and a 2-form $b_+ $ on $ M$.  See \cite{Jurco2016, gualtieri:tesi} for further details.
\end{example}

We can now connect with  one of the key notions in the formalism of~\cite{Freidel2019}.

\begin{definition} \label{compagenmetr}
A \emph{compatible generalised metric} on a para-Hermitian vector bundle \\ $(E,K,\eta)$ is a generalised metric $\mathcal{H}$ on $E$ which is compatible with the fundamental 2-form $\omega$ in the sense that
\be\nonumber
\omega^{-1}{}^\sharp\big(\mathcal{H}^\flat(e)\big) =
-\mathcal{H}^{-1}{}^\sharp\big(\omega^\flat(e)\big) \ ,  
\ee
for all $ e \in E$.
The triple $(K,\eta,\mathcal{H})$ is a \emph{Born geometry} on
$E.$ 
\end{definition}

A Born geometry can be regarded as a reduction of the structure group of $E$ to ${\sf O}(d)$, and it is  a special type of generalised metric, as asserted through~\cite{Marotta:2019eqc}

\begin{proposition}\label{prop:cHg+}
A Born geometry on an almost para-Hermitian vector bundle $(E,K,\eta)$ is a generalised metric $\cH$ specified solely by a fibrewise metric $g_+$ on the eigenbundle $L_+.$ 
\end{proposition}

In other words, a compatible generalised metric $\cH$ can be regarded
as a choice of a metric on the sub-bundle $L_+$ in the splitting $E=L_+\oplus L_-$ associated with $K.$ In this polarisation, the compatible generalised metric reads in matrix notation as\footnote{Here we consider the extension of $\eta^\flat: E \longrightarrow E^*$ to the tensor product bundles.}
\be \label{eq:diaghermmetric}
\mathcal{H} = \bigg( \begin{matrix}
g_+ & 0 \\ 0 & \eta^\flat(g_+^{-1})
\end{matrix} \bigg) \ .
\ee

\medskip

\subsection{$B_+$-Transformations}\label{sec:Btransformations} ~\\[5pt]
We describe an important class of isometries of a para-Hermitian vector bundle called `$B_+$-transformations'; as the name suggests, there is also a notion of `$B_-$-transformation', which corresponds to interchanging the roles of the eigenbundles $L_+$ and $L_-$ everywhere below.

\begin{definition}\label{def:B+}
Let $(E, K, \eta)$ be a  para-Hermitian vector bundle. A
$B_+$-\emph{transformation} is an isometry $e^{B_+}: E \longrightarrow
E$ of $\eta$ covering the identity which is given in matrix notation by
\be \label{btra}
e^{B_+}= 
\bigg(\begin{matrix}
\mathds{1} & 0 \\
B_+ & \mathds{1}
\end{matrix}\bigg)
\ee
in the  splitting $E=L_+\oplus L_-$ induced by $K$, where $B_+ : L_+
\longrightarrow L_-$ is a  skew map in the sense
that it satisfies 
$$
\eta\big(B_+(e_1),e_2\big)=- \eta\big(e_1, B_+(e_2)\big) \ ,
$$ 
for all $ e_1, e_2 \in \sfGamma(E)$.
\end{definition}

The map $B_+$ defines  a 2-form $b_+\in\sfGamma\big(\midwedge^2L_+^*\big)$ by 
$$
b_+(e_1,e_2) = \eta\big(B_+(e_1),e_2\big) \ ,
$$
for all $e_1, e_2 \in \sfGamma(E).$ 

The inverse of a $B_+$-transformation is given by the map $e^{-B_+}:E \longrightarrow E.$ The pullback of $K=
\unit_{L_+}- \unit_{L_-}\in {\sf Aut}_\unit(E)$ by a $B_+$-transformation is given by 
\be \nonumber
K_{B_+}=e^{-B_+}\circ K\circ e^{B_+}=
\begin{pmatrix}
\unit & 0 \\
-2\,B_+ & -\unit
\end{pmatrix} \ .
\ee
The $B_+$-transformation maps the polarisation
$E=L_+\oplus L_-$ induced by $K$ to a new polarisation \smash{$E=L^{B_+}_+\oplus L^{B_+}_-$}, with \smash{$L^{B_+}_\pm=e^{-B_+}(L_\pm)$}, such that only the $-1$-eigenbundle is preserved by a
$B_+$-transformation, \smash{$L^{B_+}_-=L_-,$} while the $+1$-eigenbundle changes, \smash{$L^{B_+}_+\neq L_+.$} The fundamental 2-form $\omega_{B_+}$ of the para-Hermitian structure $(K_{B_+}, \eta)$ is given by
\begin{align*}
\omega_{B_+}= \omega - 2\,b_+ \ .
\end{align*}

\begin{remark}\label{rem:B+}
If $(E,\eta,L)$ is a vector bundle of even rank
endowed with a split signature metric and a choice of maximally
isotropic sub-bundle, as in Remark~\ref{rem:EL-},  then the
maximally isotropic splittings of the short exact sequence
\eqref{eq:EL-} are mapped into each other via $B_+$-transformations which preserve $L$.
\end{remark}

It is also possible to determine the $B_+$-transformations of a
compatible generalised metric of a Born geometry. A compatible
generalised metric $\cH$ of a  para-Hermitian structure $(K,
\eta)$ transforms under a $B_+$-transformation to the compatible
generalised metric $\cH_{B_+}$ of the pullback 
para-Hermitian structure $(K_{B_+}, \eta)$ on $E.$ Recalling that
$\cH$ takes the diagonal form \eqref{eq:diaghermmetric}, an explicit calcuation then shows~\cite{Marotta:2019eqc}

\begin{proposition}
A generalised metric on a para-Hermitian
vector bundle $(E, K, \eta)$ corresponds to a choice of a Born geometry
$(K,\eta, \cH)$ and a $B_+$-transformation.
\end{proposition}

\medskip

\subsection{Metric Algebroids} ~\\[5pt]
\label{sec:malg}
The extension of type~II supergravity to the manifestly duality-invariant framework of double field theory requires a weakening of the well-known notion of Courant algebroid from generalised geometry.

\begin{definition} \label{malg}
A \emph{metric algebroid} on a manifold $ M$ is quadruple
$(E,\eta,\rho,\llbracket\,\cdot\,,\,\cdot\,\rrbracket)$ of a vector bundle
$E\longrightarrow M$ with a fibrewise non-degenerate pairing $\eta$, a vector bundle morphism $\rho: E \longrightarrow T M$ covering the identity, called the \emph{anchor}, and an $\R$-bilinear bracket $\llbracket\,\cdot \, , \,\cdot\,\rrbracket : \sfGamma(E) \times \sfGamma(E) \longrightarrow \sfGamma(E)$, called the \emph{D-bracket}, such that 
\begin{enumerate}
\item $ \llbracket e_1, f\, e_2\rrbracket = f\, \llbracket e_1, e_2\rrbracket + \big(\rho(e_1)\cdot f\big)\, e_2 \, ,$ \\[-1ex]
\item $\rho(e_1)\cdot\eta(e_2, e_3)= \eta(\llbracket e_1, e_2\rrbracket, e_3)+ \eta(e_2, \llbracket e_1, e_3\rrbracket) \, ,$ and \\[-1ex]
\item $\eta(\llbracket e_1, e_1\rrbracket, e_2)= \frac{1}{2}\, \rho(e_2)\cdot\eta(e_1, e_1)\, ,$
\end{enumerate}
for all $ e_1, e_2, e_3 \in \sfGamma(E)$ and $f \in C^{\infty}( M).$

A \emph{pre-Courant algebroid} is a metric algebroid $(E,\eta,\rho,\llbracket\,\cdot\,,\,\cdot\,\rrbracket)$ whose anchor $\rho:E\longrightarrow T M$ is a bracket morphism:
\begin{itemize}
\item[(4)] $\rho(\llbracket e_1,e_2\rrbracket) = [\rho(e_1),\rho(e_2)] \, ,$
\end{itemize}
for all $e_1,e_2\in\sfGamma(E)$, where $[\,\cdot\,,\,\cdot\,]$ is the usual Lie bracket of vector fields on $T M$. 

A \emph{Courant algebroid} is a metric algebroid $(E,\eta,\rho,\llbracket\,\cdot\,,\,\cdot\,\rrbracket)$ which is also a Leibniz-Loday algebroid, i.e. its bracket $\llbracket\,\cdot\,,\,\cdot\,\rrbracket$, called the \emph{Dorfman bracket} in this case, obeys the Jacobi identity 
\begin{itemize}
\item[(5)] $\llbracket e_1,\llbracket e_2,e_3\rrbracket\rrbracket = \llbracket\llbracket e_1,e_2\rrbracket,e_3\rrbracket + \llbracket e_2,\llbracket e_1,e_3\rrbracket\rrbracket \, ,$
\end{itemize}
for all $ e_1, e_2, e_3 \in \sfGamma(E)$.
\end{definition}

The anchored Leibniz rule (1) follows from property (2), and hence a metric algebroid can be characterised by two axioms alone. The Jacobi identity (5) together with the anchored Leibniz rule (1) implies the bracket morphism property (4). Hence, as the terminology suggests, any Courant algebroid is also a pre-Courant algebroid (but not conversely); pre-Courant and Courant algebroids can thus be specified by a minimal set of three axioms.

A metric algebroid is called \emph{regular} if its anchor $\rho$ has constant rank, and \emph{transitive} if $\rho$ is surjective. 

\begin{remark}
The anchor map $\rho:E\longrightarrow T M$ and the pairing $\eta$ induce a map $\rho^* : T^* M \longrightarrow E$ given by 
$$\eta\big(\rho^* (\alpha), e\big) \coloneqq \braket{\rho^{\tt t}(\alpha), e } \ , $$
for all $ e \in \sfGamma(E)$ and $\alpha \in \sfGamma(T^* M)$, where $\rho^{\tt t} : T^* M \longrightarrow E^*$ is the transpose of $\rho$ and $\braket{\,\cdot\,,\,\cdot\,}:\sfGamma(E^*)\times\sfGamma(E)\longrightarrow C^\infty( M)$ is the canonical dual pairing. In other words, $\rho^* = (\eta^{-1})^\sharp \circ \rho^{\tt t} .$
\end{remark}

We will now show that metric algebroids are abundant and arise very naturally. Let $(E, \eta)$ be a pseudo-Euclidean vector bundle over a manifold $ M.$ The vector bundle of metric-preserving first order covariant differential operators on $E$ forms a transitive Lie algebroid  called the \emph{Atiyah algebroid} ${\sf At}(E, \eta),$ see for instance~\cite{Mackenzie}. It fits into a short exact sequence of Lie algebroids over $ M$ given by 
\be \label{eq:atiyahalg}
0 \longrightarrow \mathfrak{so}(E) \xlongrightarrow{i} {{\sf At}}(E, \eta) \xlongrightarrow{{\sfa}} T M \longrightarrow 0 \ ,
\ee
where $\sfa$ is the  anchor of the Atiyah algebroid and $i$ is the subalgebroid inclusion of
\be \nonumber
\mathfrak{so}(E) \coloneqq \big\{ \psi \in {\sf{End}}(E) \ \big| \ \eta\circ(\psi\times\unit) = -\eta\circ(\unit\times\psi)\big\} \ .
\ee 
The isomorphism $\mathfrak{so}(E) \simeq \midwedge^2 E $ follows from this definition.

 \begin{remark} \label{metricconnections} 
Metric Koszul connections $\nabla$ on $E$ are in one-to-one correspondence with splittings of the Atiyah sequence \eqref{eq:atiyahalg}. Since a splitting of the Atiyah sequence always exists, $(E, \eta )$ always admits a metric Koszul connection.
\end{remark} 
 
\begin{remark} \label{homconnection}
Given any two metric connections $\nabla$ and $\nabla^\prime$ on $E,$ 
 from the short exact sequence \eqref{eq:atiyahalg} it follows that
\be \label{eq:diffconnection}
\nabla^\prime - \nabla = {\Omega} \ ,
\ee
for some ${\Omega} \in \mathsf{\Gamma}\big({\sf{Hom}}(T M, \mathfrak{so}(E))\big) \simeq \mathsf{\Gamma}\big(T^* M \otimes \mathfrak{so}(E)\big).$ Then ${\Omega}$ satisfies 
\be\nonumber
 \eta\big({\Omega} (X, e_1), e_2 \big) = - \eta\big( e_1, {\Omega}(X, e_2) \big) \ ,
\ee  
for all $e_1, e_2 \in \mathsf{\Gamma}(E)$ and $X \in \mathsf{\Gamma}(T M).$
\end{remark} 

We say that a vector bundle $E\longrightarrow M$ is \emph{anchored} if it is equipped with a vector bundle morphism $\rho:E\longrightarrow T M$. Following~\cite{Vaisman2012}, we then have

\begin{proposition} \label{prop:malgexistence}
Any anchored pseudo-Euclidean vector bundle $(E, \eta , \rho )$ over $ M$ admits a metric algebroid structure.
\end{proposition} 

\begin{proof}
Recall from Remark~\ref{metricconnections} that a metric connection $\nabla$ on $E$ always exists. Define the bracket operation
\be \nonumber
\llbracket\, \cdot \,,\, \cdot\, \rrbracket_\nabla \colon \mathsf{\Gamma}(E) \times \mathsf{\Gamma}(E) \longrightarrow \mathsf{\Gamma}(E)
\ee
by
\be \label{eq:nablabracket}
\eta(\llbracket e_1 , e_2 \rrbracket_\nabla,e) \coloneqq \eta(\nabla_{\rho(e_1)} e_2 - \nabla_{\rho(e_2)}e_1,e) +  \eta( e_1, \nabla_{\rho(e)} e_2 )  \ ,
\ee
for all $e, e_1, e_2 \in  \mathsf{\Gamma}(E).$ From the definition \eqref{eq:nablabracket} we find
\begin{align}\nonumber
 \eta(\llbracket e,e_1\rrbracket_\nabla, e_2) + \eta( e_1, \llbracket e, e_2\rrbracket_\nabla) = \eta( \nabla_{\rho(e)} e_1, e_2 ) + \eta( e_1, \nabla_{\rho(e)}e_2 ) = \rho(e)\cdot\eta( e_1, e_2) \  , 
\end{align}
which proves property (2) of Definition~\ref{malg}. Similarly
\be \nonumber
\eta( \llbracket e , e \rrbracket_\nabla , e_1 ) =  \eta( e, \nabla_{\rho(e_1)} e) = \tfrac12\,  \rho(e_1)\cdot\eta( e, e) \  ,
\ee
and hence the bracket $\llbracket\, \cdot \,,\, \cdot \, \rrbracket_\nabla$ satisfies property (3) of Definition~\ref{malg}.
\end{proof}

We now show that there are in fact infinitely many metric algebroid structures on any anchored pseudo-Euclidean vector bundle $(E, \eta , \rho ).$ This follows from~\cite{Vaisman2012}

\begin{proposition}  \label{Dbracket3form}
Let $(E, \eta , \rho )$ be an anchored pseudo-Euclidean vector bundle over $ M$ endowed with two  metric connections $\nabla^\prime$ and $\nabla$. Then the difference between their induced D-brackets is a $3$-form $F \in \mathsf{\Gamma}(\midwedge^3 E^*)$ given by
\be \label{eq:dbracket3form}
F(e_1, e_2, e) = \eta( \llbracket e_1, e_2 \rrbracket_{\nabla^\prime} - \llbracket e_1, e_2 \rrbracket_{\nabla}, e ) \ , 
\ee
for all $e, e_1, e_2 \in \mathsf{\Gamma}(E).$
\end{proposition}

\begin{proof}  
By Remark \ref{homconnection}, the difference ${\Omega}$ between the two connections $\nabla^\prime$ and $\nabla$ induces 
\be\nonumber
{\Lambda}= {\Omega} \circ \rho \ \in \ \mathsf{\Gamma}\big(\mathsf{Hom}(E, \mathfrak{so}(E))\big) \ ,
\ee
such that
\be\nonumber
\eta\big( {\Lambda} (e, e_1), e_2 \big)= - \eta\big( e_1, {\Lambda}(e, e_2) \big) \ ,
\ee  
for all $e, e_1, e_2 \in \mathsf{\Gamma}(E).$
By substituting \eqref{eq:diffconnection} into \eqref{eq:nablabracket} it follows that
\be\nonumber
\eta( \llbracket e_1, e_2\rrbracket_{\nabla^\prime}, e ) = \eta( \llbracket e_1, e_2\rrbracket_{\nabla}, e ) + \eta\big( {\Lambda}(e_1, e_2)- {\Lambda}(e_2, e_1), e \big) + \eta\big( e_2, {\Lambda}(e, e_1)\big) \ .
\ee
Define the 3-form ${F} \in \mathsf{\Gamma}(\midwedge^3 E^*)$ by 
\be\nonumber
{F}(e_1, e_2, e)=  \eta\big( {\Lambda}(e_1, e_2) - {\Lambda}(e_2, e_1), e \big) + \eta\big( e_2, {\Lambda}(e, e_1)\big)
\ee
and then \eqref{eq:dbracket3form} follows.
\end{proof}  

\begin{remark}  \label{Correspondence3forms}
Proposition \ref{Dbracket3form} implies that there is a one-to-one correspondence between metric algebroid structures on $(E, \eta , \rho )$  and 3-forms ${F} \in \mathsf{\Gamma}(\midwedge^3 E^*).$ This correspondence is not canonical and depends on the choice of a reference D-bracket on $(E, \eta , \rho ).$
\end{remark}  

\begin{example}\label{ex:metricalgeeta}
Let $( M,\eta)$ be a pseudo-Riemannian manifold, and let $\nabla^\LC$ denote the Levi-Civita connection of the metric $\eta$. Define a bracket operation $\llbracket\,\cdot\,,\,\cdot\,\rrbracket_\LC:\mathsf{\Gamma}(T M)\times\mathsf{\Gamma}(T M)\longrightarrow\mathsf{\Gamma}(T M)$ by
\be\nonumber
\eta(\llbracket X,Y\rrbracket_\LC,Z) = \eta(\nabla^\LC_XY-\nabla_Y^\LC X,Z)  + \eta(\nabla_Z^\LC X,Y)
\ee
for vector fields $X,Y,Z\in\mathsf{\Gamma}(T M)$. Then $(T M,\eta,\unit_{T M},\llbracket\,\cdot\,,\,\cdot\,\rrbracket_\LC)$ is a regular metric algebroid.
\end{example}

\begin{example}\label{ex:canmetricalg}
On any almost para-Hermitian manifold $(M,K,\eta)$ one can construct the regular \emph{canonical metric algebroid} $(TM,\eta,\unit_{TM},\llbracket\,\cdot\,,\,\cdot\,\rrbracket_{\tt can})$ in the following way (see~\cite{Freidel2017,Svoboda2018,Freidel2019}). Let $\sfp_\pm:TM\longrightarrow L_\pm$ be the projections to the $\pm\,1$-eigenbundles $L_\pm$ of the almost para-complex structure $K$, and define the \emph{canonical connection}
\begin{align*}
\nabla^{\tt can} := \sfp_+\circ\nabla^\LC\circ\sfp_+ + \sfp_-\circ\nabla^\LC\circ\sfp_- \ ,
\end{align*}
where $\nabla^\LC$ is the Levi-Civita connection of the split signature metric $\eta$. Then $\nabla^{\tt can}$ also preserves $\eta$, and by definition it preserves $K$ as well, i.e. $\nabla^{\tt can}_X K=0$ for all $X\in\sfGamma(TM)$, hence it preserves the splitting $TM = L_+\oplus L_-$ induced by $K$. The deviation between the two linear connections $\nabla^{\tt can}$ and $\nabla^\LC$ (which by Remark~\ref{metricconnections} correspond to different splittings of the corresponding Atiyah sequence \eqref{eq:atiyahalg} with $E=TM$) can be quantified through
\begin{align*}
\eta(\nabla^{\tt can}_X \, Y, Z)= \eta(\nabla^{\textrm{\tiny\tt LC}}_X \, Y, Z) -\tfrac12\, \nabla^{\textrm{\tiny\tt LC}}_X\, \omega(Y, K(Z)) \ ,
\end{align*}
for all $X,Y,Z\in\sfGamma(TM)$,
where $\omega$ is the fundamental $2$-form of the almost para-Hermitian manifold. In particular, $\nabla^{\tt can} = \nabla^\LC$ if and only if $(M,K,\eta)$ is an \emph{almost para-K\"ahler manifold}, i.e. $\de\omega=0$~\cite{Svoboda2018}.

By the construction of Proposition~\ref{prop:malgexistence}, we obtain a D-bracket on $TM$ given by
\begin{align*}
\eta(\llbracket X,Y\rrbracket_{\tt can},Z) := \eta(\nabla_X^{\tt can}Y - \nabla_Y^{\tt can}X,Z) + \eta(\nabla_Z^{\tt can}X,Y) \ ,
\end{align*}
for all $X,Y,Z\in\sfGamma(TM)$. It is compatible with the almost para-complex structure $K$, in the sense that both of the sub-bundles $L_\pm$ are involutive with respect to it: $$\llbracket\sfGamma(L_\pm),\sfGamma(L_\pm)\rrbracket_{\tt can} \ \subseteq \ \sfGamma(L_\pm) \ . $$ It is moreover compatible with the Lie bracket of vector fields on $TM$, in the sense that
\begin{align*}
\llbracket \sfp_\pm(X),\sfp_\pm(Y)\rrbracket_{\tt can} = \sfp_\pm\big([\sfp_\pm(X),\sfp_\pm(Y)]\big) \ ,
\end{align*}
for all $X,Y\in\sfGamma(TM)$. In this sense, the bracket $\llbracket\,\cdot\,,\,\cdot\,\rrbracket_{\tt can}$ is the unique canonical compatible D-bracket induced by a metric connection on the anchored para-Hermitian vector bundle $(TM,\eta,K,\unit_{TM})$~\cite{Freidel2017}. 
\end{example}

\medskip

\subsection{Exact Pre-Courant Algebroids}\label{sec:expreCourant} ~\\[5pt]
We shall now specialise the discussion of Appendix~\ref{sec:malg} to the class of metric algebroids which are pre-Courant algebroids, i.e. for which the homomorphism property (4) of the anchor map in Definition~\ref{malg} is imposed. Let $(E,\eta,\rho,\llbracket\,\cdot\,,\,\cdot\,\rrbracket)$ be a regular pre-Courant algebroid over a manifold $ M$. From the definition of $\rho^*=\eta^{-1}{}^\sharp \circ \rho^{\tt t}:T^* M\longrightarrow E$ it follows that  ${\rm im}(\rho^*)$ is a coisotropic sub-bundle of $E$ with respect to the fibrewise metric $\eta$, and $\sfGamma({\rm im}(\rho^*))$ is a two-sided abelian ideal of $\sfGamma(E)$ with respect to the D-bracket  $\llbracket\,\cdot \, , \, \cdot\,\rrbracket$~\cite{Kotov:2010wr}. Hence $\rho \circ \rho^* =0$ and there is a chain complex
\begin{align}\label{eq:CAexact}
0\longrightarrow T^*M \xlongrightarrow{\rho^*} E \xlongrightarrow{\rho} TM \longrightarrow 0 \ .
\end{align}
A pre-Courant algebroid is \emph{exact} if \eqref{eq:CAexact} is a short exact sequence.  

From  the exactness of the 
sequence \eqref{eq:CAexact}, it follows that the sub-bundle
$\mathrm{im}(\rho^*) \subset E,$ which is isomorphic to $T^* M,$ is maximally isotropic with respect to $\eta.$
A para-Hermitian structure on $E$ is given by a choice of a maximally
isotropic splitting of \eqref{eq:CAexact}:
$$
\sigma: T M \longrightarrow E \qquad \mbox{with} \quad \rho \circ \sigma=
\unit_{T M} \ . 
$$
It follows that 
$$
E= \mathrm{im}(\sigma) \oplus \mathrm{im}(\rho^*)
$$ 
with associated para-complex structure defined by
$$
K_\sigma\big(\sigma(X)+ \rho^*(\alpha)\big)=\sigma(X)- \rho^*(\alpha) \ , 
$$ 
for all $ X \in \mathsf{\Gamma}(T M)$ and $\alpha \in
{\mathsf\Omega}^1(M).$  The para-complex structure $K_\sigma$ is
compatible with the metric $\eta,$ and thus $E$ is endowed with a
para-Hermitian structure. This para-Hermitian structure of an exact
pre-Courant algebroid is isomorphic to the para-Hermitian structure of the
generalised tangent bundle $\mathbb{T} M$ from
Example~\ref{ex:gentanbun}.

Under the isomorphism $E\simeq\IT M = T M\oplus T^* M$ induced by the splitting $\sigma$, the anchor map $\rho:E\longrightarrow T M$ is sent to the projection ${\rm pr}_{T M}:T M\oplus T^* M \longrightarrow T M$ defined by ${\rm pr}_{T M}(X+\alpha) = X$, for all $X\in T M$ and $\alpha\in T^* M$. The D-bracket $\llbracket\,\cdot\,,\,\cdot\,\rrbracket$ is sent to the bracket~\cite{gualtieri:tesi,Severa-letters}
\be\label{eq:standardbracket}
\llbracket X+ \alpha, Y + \beta\rrbracket_{H_\sigma} = [X,Y]+ \pounds_X \beta  -\iota_Y \, \de \alpha + \iota_Y\iota_X H_\sigma \ ,
\ee
for all $X,  Y \in \sfGamma(T M)$ and $\alpha, \beta \in {\mathsf\Omega}^1(M)$, where $\pounds_X$ denotes the Lie derivative along the vector field $X$ and $H_\sigma\in{\mathsf\Omega}^3( M)$ is the $3$-form on $ M$ defined by 
\begin{align}\label{eq:Hsigma}
H_\sigma(X,Y,Z)= \eta\big(\llbracket\sigma(X),\sigma(Y)\rrbracket, \sigma(Z)\big) \ ,
\end{align}
for all $X,Y,Z \in \sfGamma(T M).$ 

By Remark~\ref{rem:EL-}, every maximally isotropic splitting of the short exact sequence \eqref{eq:CAexact} is associated with a para-Hermitian structure on $E$, and by Remark~\ref{rem:B+} any two splittings of an exact pre-Courant algebroid are related by a $B_+$-transformation, as discussed in Appendix~\ref{sec:Btransformations}. Under the isomorphism above, a $B_+$-transformation of a split exact pre-Courant algebroid is given by a $2$-form $b_+=B\in{\mathsf\Omega}^2( M)$ on $ M$ which induces the vector bundle morphism $B_+=B^\flat:T M\longrightarrow T^* M$ given by  $B^\flat(X) = \iota_XB$, for all $X\in\sfGamma(T M)$. Thus $e^{B_+}:\IT M\longrightarrow\IT M$ is given by
\begin{align*}
e^{B_+}(X+\alpha) = X+\iota_XB+\alpha \ ,
\end{align*}
for all $X\in\sfGamma(T M)$ and $\alpha\in{\mathsf\Omega}^1(M)$. The image ${\rm im}(e^{B_+})\subset\IT M$ is a maximally isotropic sub-bundle, and it is compatible with the anchor ${\rm pr}_{T M}$, i.e. ${\rm pr}_{T M}\circ e^{B_+} = e^{B_+}\circ{\rm pr}_{T M}$. The isometry $e^{B_+}$ acts on the D-bracket \eqref{eq:standardbracket} as~\cite{gualtieri:tesi}
\begin{align*}
\llbracket e^{B_+}(X+\alpha),e^{B_+}(Y+\beta)\rrbracket_{H_\sigma} = e^{B_+}\big(\llbracket X+\alpha,Y+\beta\rrbracket_{H_\sigma+\de B}\big) \ ,
\end{align*}
for all $X+\alpha,Y+\beta\in\sfGamma(\IT M)$. On the other hand, $B_+$-transformations do not preserve the para-Hermitian structure on $\IT M$, but the different para-Hermitian structures are isomorphic, and they have the same D-bracket if and only if the $2$-form $B$ is closed, $\de B=0$. 

\begin{remark}
When $H_\sigma=0$, the para-Hermitian vector bundle $\IT M$ from Example~\ref{ex:gentanbun} with anchor ${\rm pr}_{T M}$ and D-bracket \eqref{eq:standardbracket} is a Courant algebroid called the \emph{standard Courant algebroid} on $ M$. In this case the standard Dorfman bracket $\llbracket\,\cdot\,,\,\cdot\,\rrbracket_0 = \llbracket\,\cdot\,,\,\cdot\,\rrbracket_\nabla$ can also be obtained from the construction \eqref{eq:nablabracket} by choosing an arbitrary torsion-free connection $\nabla$ on $T M$ and extending it to $\IT M$ by $\nabla_X(Y+\beta) = \nabla_XY+\nabla_X\beta$, for all $X\in\sfGamma(T M)$ and $Y+\beta\in\sfGamma(\IT M)$~\cite{Svoboda2018}. This bracket is compatible with the para-Hermitian structure of Example~\ref{ex:gentanbun}, i.e. the eigenbundles $T M$ and $T^* M$ are both involutive with respect to \eqref{eq:standardbracket} for $H_\sigma=0$. 

More generally, the D-bracket is a Dorfman bracket, i.e. it satisfies the Jacobi identity~(5) of Definition~\ref{malg}, if and only if the $3$-form $H_\sigma$ is closed, $\de H_\sigma=0$. In that case the Courant algebroid is called the \emph{$H_\sigma$-twisted standard Courant algebroid} on~$ M$, and its isomorphism class depends only on the de~Rham cohomology class $[H_\sigma]\in{\sf H}^3( M;\IR)$, called the \emph{\v{S}evera class}. In other words, isomorphism classes of exact Courant algebroids over a manifold $ M$ are in one-to-one correspondence with elements of the cohomology group ${\sf H}^3( M;\IR)$. This is the famous \v{S}evera classification of exact Courant algebroids~\cite{Severa-letters}.
\end{remark}

\medskip

\subsection{Courant Algebroid Reduction} ~\\[5pt]
\label{app:Courantred}
We briefly recall the main results concerning the reduction of Courant algebroids from~\cite{Bursztyn2007, Zambon2008}, beginning with the following notion.

\begin{definition} \label{def:basic}
Let $(E, \eta , \rho, \llbracket \, \cdot \, , \, \cdot \, \rrbracket)$ be an exact Courant algebroid over a manifold $M$ endowed with an isotropic sub-bundle $A\subset E$ over a submanifold $\cW \subseteq M$ whose orthogonal complement $A^\perp$ obeys $\rho(A^\perp)=T\cW.$ The subspace of sections of $A^\perp$ which are \emph{basic with respect to} $A$ is 
\be \nonumber
\mathsf{\Gamma}_{\mathtt{bas}}(A^\perp) \coloneqq \big\{ e \in \mathsf{\Gamma}(A^\perp) \ \big| \ \llbracket e , \sfGamma(A) \rrbracket \subseteq \mathsf{\Gamma}(A) \big\} \ .
\ee
\end{definition}

The circumstances under which the reduction of an exact Courant algebroid is possible is described by

\begin{theorem} \label{thm:BCGZ}
Let $(E, \eta , \rho, \llbracket \, \cdot \, , \, \cdot \, \rrbracket)$ be an exact Courant algebroid over $M$  endowed with an isotropic sub-bundle $A$ over a submanifold $\cW \subseteq M$ such that $\rho(A^\perp)=T\cW$. Assume that $A^\perp$ is spanned fibrewise by basic sections
$\mathsf{\Gamma}_{\mathtt{bas}}(A^\perp).$ Let $\cF$ be the foliation of $\cW$ integrating the involutive distribution $\rho(A) \subset T\cW$, and suppose that the leaf space $\cQ= \cW / \cF$ is a smooth manifold.  Then there is an exact Courant algebroid with underlying vector bundle $E_{\mathtt{red}}$ over $\cQ$ such that the diagram 
\be \nonumber
\begin{tikzcd}
A^\perp /A \arrow{r}{} \arrow{dd} & E_{\mathtt{red}} \arrow{dd}{} \\ & \\
\cW \arrow{r}{} & \cQ
\end{tikzcd}
\ee
is a pullback of vector bundles.
\end{theorem}

The main result concerning the reduction of Dirac structures can be stated as

\begin{proposition} \label{prop:Dircred}
Let $(E, \eta , \rho, \llbracket \, \cdot \, , \, \cdot \, \rrbracket)$ be an exact Courant algebroid over $M$  endowed with an isotropic sub-bundle $A$ over a submanifold $\cW \subseteq M$ such that $\rho(A^\perp)=T\cW$, which fulfils the assumptions of Theorem \ref{thm:BCGZ}, i.e. 
$E$ induces an exact Courant algebroid $E_{\mathtt{red}}$ over $\cQ= \cW /\cF.$ Let $L$ be a maximally isotropic sub-bundle of $E \rvert_\cW$ such that $L \cap A^\perp$ has constant rank and 
\be \label{eq:diracred1}
\llbracket \mathsf{\Gamma}(A), \mathsf{\Gamma}(L \cap A^\perp) \rrbracket \  \subseteq \ \mathsf{\Gamma}(L+A) \ .
\ee
Then $L$ reduces to a maximally isotropic sub-bundle $L_{\mathtt{red}}$ of $E_{\mathtt{red}}.$
If in addition
\be \label{eq:diracred2}
\llbracket \mathsf{\Gamma}_{\mathtt{bas}}(L \cap A^\perp), \mathsf{\Gamma}_{\mathtt{bas}}(L \cap A^\perp) \rrbracket \ \subseteq \ \mathsf{\Gamma}(L+A) \ ,
\ee
then $L_{\mathtt{red}}$ is a Dirac structure.
\end{proposition}

\newpage


\end{document}